\documentclass[hidelinks,onefignum,onetabnum]{siamart251104}

\usepackage{lipsum}
\usepackage{amsfonts}
\usepackage{graphicx}
\usepackage{subcaption}
\usepackage{float}
\usepackage{mathtools}
\usepackage{epstopdf}
\usepackage{algorithmic}
\ifpdf
  \DeclareGraphicsExtensions{.eps,.pdf,.png,.jpg}
\else
  \DeclareGraphicsExtensions{.eps}
\fi

\newsiamremark{remark}{Remark}
\newtheorem{example}[theorem]{Example}
\newsiamremark{hypothesis}{Hypothesis}
\crefname{hypothesis}{Hypothesis}{Hypotheses}
\newsiamthm{claim}{Claim}
\newsiamremark{fact}{Fact}
\crefname{fact}{Fact}{Facts}
\newtheorem{question}{Question}
\newcommand{\defword}[1]{\emph{#1}}
\allowdisplaybreaks[4]

\headers{Minimal Degenerate Zero-One Reaction Networks}{Y. Chen and X. Tang}

\title{Characterization of Minimal Degenerate Zero-One Reaction Networks\thanks{Submitted to the editors DATE.
}}

\author{
Yuanlin Chen
\and
Xiaoxian Tang\thanks{School of Mathematical Sciences, Beihang University, Beijing, China (\email{xiaoxian@buaa.edu.cn}).}
}

\ifpdf
\hypersetup{
  pdftitle={Characterization of Minimal Degenerate Zero-One Reaction Networks},
  pdfauthor={Y. Chen and X. Tang}
}
\fi

\begin{document}

\maketitle

\begin{abstract}
A fundamental problem in the algebraic study of biochemical reaction networks is to characterize degeneracy. Two-dimensional zero-one networks, where each reactant appears with stoichiometric coefficient zero or one, constitute the smallest biologically relevant class capable of exhibiting degeneracy. In this paper, we provide a complete classification: a two-dimensional zero-one network without trivial species is degenerate if and only if it is a consistent subnetwork of a species refinement of one of two prototypical networks, namely the complete paired-exchange network or the catalytic-pair conversion network. Equivalently, in matrix-theoretic terms, degeneracy is determined entirely by the row patterns of the stoichiometric and reactant matrices, and can be verified by purely structural inspection without computation. Remarkably, every such degenerate network has a steady-state system consisting entirely of binomials, so its steady-state variety is toric. These results also bear on absolute concentration robustness: since any network exhibiting this property necessarily contains a degenerate subnetwork, the minimal degenerate networks characterized here serve as fundamental building blocks for constructing such networks.
\end{abstract}

\begin{keywords}
stoichiometric matrix, reactant matrix, degeneracy, biochemical reaction network
\end{keywords}

\begin{MSCcodes}
15A21, 05C50, 92C42
\end{MSCcodes}

\section{Introduction}
\label{sec1}

For dynamical systems arising from biochemical reaction networks, nondegeneracy is central to the study of complex dynamical behaviors. 
A network  is degenerate  if every positive steady state is degenerate. In contrast, a network is nondegenerate if it admits at least one nondegenerate positive steady state. Deciding whether a network is degenerate or nondegenerate is essentially the same problem. On the one hand, the existence of nondegenerate positive steady states is often a key prerequisite for phenomena such as multistability and oscillations \cite{BanajiBorosHofbauer2024,DomijanKirkilionis2009,JiaoTangZeng2025}. On the other hand, in cellular regulation, fundamental biological processes, including signal transduction, cellular decision-making, and gene-expression control, are closely linked to important dynamical properties of the underlying network \cite{ConradiShiu2018,JaniakSpens2005}, such as nondegeneracy \cite{Mueller2016,ShiuDeWolff2019}, robustness \cite{Puente2025,ShinarAlonFeinberg2009}, and bifurcations \cite{BanajiBoros,BanajiBorosHofbauer2025,TangWang2023}. Therefore, the study of nondegeneracy not only helps clarify the mechanisms underlying complex dynamics but also reveals the structural basis of biological function. For this reason, we propose the following question.
\begin{question}
\label{q1}
What structure characterizes degenerate networks and makes it possible to determine nondegeneracy explicitly, without computation?
\end{question}

We now describe the connection between the degeneracy criterion and matrix theory. A reaction network is uniquely determined by its stoichiometric matrix $\mathcal{N}$ and reactant matrix $\mathcal{X}$, and thus can be identified with the pair $(\mathcal{N},\mathcal{X})$. Through this identification, a standard algebraic criterion translates degeneracy into a condition on the matrix $A:=\mathcal{N}\operatorname{diag}(\gamma)\mathcal{X}^\top$, where $\gamma$ is any element of the flux cone $\mathcal{F}(\mathcal{N})$. Specifically, the network is degenerate if and only if every $r\times r$ principal minor of $A$ vanishes identically, with $r:=\operatorname{rank}(\mathcal{N})$. Consequently, classifying all degenerate networks is equivalent to characterizing every matrix pair $(\mathcal{N},\mathcal{X})$ for which this rank-deficiency condition holds. 
We first illustrate the algebraic criterion through the following network, postponing the formal definitions:
\begin{align}
\label{eq:net1}
X_1+X_2 \xrightarrow{\kappa_1} X_2+X_3,~
X_3 \xrightarrow{\kappa_2} X_1+X_2,~
X_1+X_2 \xrightarrow{\kappa_3} X_1.
\end{align}
This network uniquely corresponds to the following pair of matrices,
\[
\mathcal{N}=
\begin{pmatrix}
-1 & 1 & 0\\[2pt]
0 & 1 & -1\\[2pt]
1 & -1 & 0
\end{pmatrix},
\qquad
\mathcal{X}=
\begin{pmatrix}
1 & 0 & 1\\[2pt]
1 & 0 & 1\\[2pt]
0 & 1 & 0
\end{pmatrix}.
\]
The flux cone $\mathcal{F}(\mathcal{N})$ is determined by 
$\mathcal{N}\gamma=\mathbf{0}$ with $\gamma\ge\mathbf{0}$, which gives 
$\gamma_1=\gamma_2=\gamma_3$. Hence 
$\mathcal{F}(\mathcal{N})=\{\lambda(1,1,1)^\top\mid\lambda\ge0\}$ and
\[
A(\lambda):=\mathcal{N}\begin{pmatrix}
\lambda & 0 & 0\\
  0 & \lambda & 0\\
  0 & 0 & \lambda
\end{pmatrix}
\mathcal{X}^\top
=\lambda
\begin{pmatrix}
-1 & -1 & 1\\
-1 & -1 & 1\\
1 & 1 & -1
\end{pmatrix}.
\]Note that $\mathcal{N}$ has rank $2$. 
 One readily checks that 
all $2\times 2$ principal minors of $A(\lambda)$ are identically zero. It follows that the given network is degenerate.

The study of degeneracy is central to understanding exceptional dynamical behaviors in biochemical systems. Many realistic networks contain only degenerate subnetworks, motivating methods to regularize them for extending classical chemical reaction network theory \cite{UhrKaltenbachConradiStelling2009}. Degenerate networks also represent boundary cases in algebraic geometry where the positive zero set has a nonempty but nowhere-dense parameter region \cite[Theorem~3.4]{Feliu2026}. In reaction-diffusion systems, such degeneracy significantly influences convergence to steady states, with entropy-based methods providing robust tools for handling nonlinearities of arbitrary order \cite{DesvillettesPhungTang2026}.

Conversely, nondegeneracy underpins the inheritance of complex dynamics: nondegenerate multistationarity, multistability, and Hopf bifurcations can be lifted from subnetworks to the whole network \cite{Banaji2018,Banaji2023,BanajiPantea2018,JoshiShiu2013}, preserved under species additions when rank is unchanged \cite{BanajiBorosHofbauer2022}, and maintained under network enlargements satisfying a general transversality condition \cite{BanajiBorosHofbauer2025}. These inheritance properties have motivated the complete characterization of small networks admitting multistability or Hopf bifurcations \cite{JoshiShiu2017,McClureShiu2024,TangWang2023,TangXu2021}, as well as algorithmic methods for detecting Hopf bifurcations via stoichiometric analysis combined with computational logic and tropical geometry \cite{ErramiEiswirthGrigorievEtAl2015}. A common and critical first step in these algorithms is to exclude degenerate networks.  

However, despite its importance, an explicit structural criterion for degeneracy remains 
largely unavailable, particularly within the fundamental class of zero-one networks (where each reactant appears at most once per reaction, as is typical in many biochemical systems). It is known that one-dimensional zero-one networks admit either no positive steady state or a unique nondegenerate one (Theorem~2 in \cite{JiaoTangZeng2025}). So, the two-dimensional case is the first nontrivial setting in which degeneracy can occur. More recently, all two-dimensional zero-one networks with up to three species were computationally enumerated and all degenerate cases were identified \cite{TangWangZhang2026}. It was observed that the steady-state system of each such degenerate network consists of two monomials, but this binomial property was only noted empirically and lacks a general mathematical proof. For instance, recall the network \eqref{eq:net1}. Under mass-action kinetics, the concentration dynamics are governed by
\begin{align*}
\dot{x}_1=-\kappa_1x_1x_2+\kappa_2x_3,\
\dot{x}_2=\kappa_2x_3-\kappa_3x_1x_2,\
\dot{x}_3=\kappa_1x_1x_2-\kappa_2x_3.
\end{align*}
Clearly, all steady-state equations above have a binomial structure; equivalently, the steady-state variety is toric. 

The main contribution of this paper is a complete answer to Question~\ref{q1} for two-dimensional zero-one networks. We prove that a two-dimensional zero-one network without trivial species is degenerate if and only if it is a consistent subnetwork of a species refinement of one of two prototypical networks: the complete paired-exchange network or the catalytic-pair conversion network (Theorem~\ref{thm:main}). Equivalently, in matrix-theoretic terms, this means that every row of the stoichiometric and reactant matrices must belong to one of two finite sets of explicit row patterns derived from these prototypes (Theorem~\ref{thm:deg-char}). Consequently, for such networks, degeneracy can be determined by purely structural inspection without any computation. As a further consequence, the steady-state system of every degenerate two-dimensional zero-one network identified by our characterization consists entirely of binomials; hence its steady-state variety admits a toric structure.

Beyond its intrinsic theoretical interest, such an explicit characterization of degeneracy has direct algorithmic consequences. In particular, it is instrumental for the detection of absolute concentration robustness (ACR). Recent results show that for full-dimensional networks (i.e., those without conservation laws), a rank-based criterion \cite[Theorem~B]{Feliu2026} establishes that a necessary condition for a species to exhibit ACR is that removing this species yields a degenerate lower-dimensional network. For networks with conservation laws, if removing a species produces a nondegenerate subnetwork, then that species cannot have ACR \cite[Theorem~3.1]{si2026}. Together, these results reveal that degenerate subnetwork structures are inherently present in ACR networks. Consequently, an explicit criterion for degeneracy provides an efficient tool for determining ACR.

The significance of the binomial structure uncovered in our characterization extends even further. For linearly binomial networks, which are common in applications, multistationarity is determined by inspecting the critical function, and the resulting parameter regions are full-dimensional \cite{DickensteinMillanShiuEtAl2019}. This property is deeply connected to toric networks (complex-balanced mass-action systems), whose toric loci admit a rich geometric structure including product decomposition, contractibility, and affine invariance \cite{CraciunJinSorea2025}. These toric ideas have further been extended to general polynomial dynamical systems via Euclidean embedded graphs and toric differential inclusions, with applications to persistence and the Global Attractor Conjecture \cite{Craciun2019}. More broadly, toric varieties are among the most classical and well-studied objects in algebraic geometry, providing a fundamental bridge between commutative algebra, combinatorics, and convex geometry \cite{CoxLittleSchenck2011}. Their significance extends far beyond reaction network theory, with deep applications in mirror symmetry \cite{ChanLauLeungTseng2017}, optimization, coding theory, and algebraic statistics \cite{CoxLittleSchenck2011}. Our characterization thus situates the dynamics of degenerate chemical networks within this rich and classical mathematical framework.

The remainder of this paper is organized as follows. In Section~\ref{sec2}, we review the basic concepts and notation used throughout the paper. In Section~\ref{sec3}, we present the main result for two-dimensional zero-one networks, including both the network version (Theorem~\ref{thm:main}) and the matrix version (Theorem~\ref{thm:deg-char}), followed by examples illustrating applications to the detection of non-vacuous ACR networks of higher dimension. In Section~\ref{sec4}, we provide the proofs of the main results. Specifically, in Section~\ref{sec4.1}, we characterize the two-species case (Theorem~\ref{thm:2s}), and in Section~\ref{sec4.2}, we treat the general case with an arbitrary number of species. Finally, in Section~\ref{sec5}, we conclude with a discussion and outline the direction for future research.

\section{Background}
\label{sec2}
In this section, we provide background on chemical reaction networks (Section \ref{sec2.1}) and the standard criteria for degeneracy based on linear algebra (Section \ref{sec2.2}). The reader is referred to \cite{ConradiFeliuMincheva2017,JiaoTangZeng2025,TangWangZhang2026} for further details.
\subsection{Reaction network}\label{sec2.1}

A \defword{reaction network} $G$, or simply a \defword{network}, consists of a set of $s$ species $\{X_1, \dots, X_s\}$ and a set of $m$ reactions of the form
\begin{align}
\label{eq:network}
\alpha_{1j}X_1 +
 \dots +
\alpha_{sj}X_s
~ \xrightarrow{\kappa_j} ~
\beta_{1j}X_1 +
 \dots +
\beta_{sj}X_s,
 \quad 
    \text{for } j=1, \ldots, m,
\end{align}
where $\alpha_{ij}$ and $\beta_{ij}$ are nonnegative integers called \defword{stoichiometric coefficients}. We assume that $(\alpha_{1j},\ldots,\alpha_{sj}) \neq (\beta_{1j},\ldots,\beta_{sj})$ for every reaction. In addition, each $\kappa_j \in \mathbb{R}_{>0}$ denotes the \defword{rate constant} of the $j$-th reaction in \eqref{eq:network}. A \defword{zero-one network} is a reaction network that consists solely of reactions whose stoichiometric coefficients are all $0$ or $1$. The \defword{stoichiometric matrix} $\mathcal{N}$ and \defword{reactant matrix} $\mathcal{X}$  of $G$ are the $s\times m$ matrices with $(i,j)$-entries $\beta_{ij}-\alpha_{ij}$ and $\alpha_{ij}$, respectively.  We denote the $i$-th row of $\mathcal{N}$ by $\mathcal{N}_i$, and the $j$-th column of $\mathcal{N}$ by $\mathrm{col}_j(\mathcal{N})$. The \defword{stoichiometric subspace}, denoted by $S$, is the real linear subspace spanned by the column vectors $col_1(\mathcal{N}),\ldots,col_m(\mathcal{N})$ of $\mathcal{N}$. The \defword{dimension of the network} $G$ is defined by the dimension of $S$, i.e., $\operatorname{rank}(\mathcal{N})$. 
A species of a zero-one network is called \defword{trivial} if, in every reaction, it either does not appear or appears on both sides of the reaction \cite[Section 2.1]{BanajiBorosHofbauer2024}.  Notice that the row corresponding to a trivial species in the stoichiometric matrix is identically zero. Consequently, removing trivial species yields a network that is dynamically equivalent to the original one. Therefore, we typically focus only on networks without trivial species.

Note that a network $G$ is uniquely determined by its stoichiometric matrix 
$\mathcal{N}$ and reactant matrix $\mathcal{X}$. For simplicity, we also represent 
a network $G$ by the pair $(\mathcal{N},\mathcal{X})$. A network $G'$ is 
\defword{isomorphic} to another network $G$ if $G'$ can be obtained from $G$ 
by independently relabeling the species or the reactions 
\cite[Definition 2.1]{TangXu2021}. Equivalently, $G$ and $G'$ are isomorphic 
if and only if there exist permutation matrices $P$ and $Q$ such that
\[
\mathcal{N}' = P\mathcal{N}Q \quad\text{and}\quad \mathcal{X}' = P\mathcal{X}Q.
\]
In what follows, for brevity, we say that a network $G$ is $G'$ (or is 
$({\mathcal N}', {\mathcal X}')$) to mean that $G$ is isomorphic to $G'$ 
(or to the network uniquely determined by $({\mathcal N}', {\mathcal X}')$).

For any $r$-dimensional network $G$, a \defword{subnetwork} of $G$ is an $r$-dimensional network obtained by removing some reactions from $G$.
We remark that $G$ itself is also  a subnetwork of $G$.

The concentrations of the species $X_1, \ldots, X_s$ are represented by $x_1, \ldots, x_s$, respectively. Under mass-action kinetics, the time evolution of the species concentrations is governed by the following system of ODEs:
\begin{align}\label{eq:sys}
\dot{x} = f(\kappa,x) := \mathcal{N}v(\kappa,x),
\end{align}
where $\kappa:=(\kappa_1,\ldots,\kappa_m)^\top$, 
$x:=(x_1,\ldots,x_s)^\top$, 
and
$v:=(v_1,\ldots,v_m)^\top$ with
$v_j(\kappa,x):=\kappa_j\prod\limits_{i=1}^s x_i^{\alpha_{ij}}$
for $j=1,\ldots,m$.

Let $d := s - \operatorname{rank}(\mathcal{N})$. A $d \times s$ matrix $W$ is called a 
\defword{conservation-law matrix} if its rows form a basis for the left null space of 
$\mathcal{N}$, i.e., $W\mathcal{N} = \mathbf{0}$. Consequently, along any trajectory of 
\eqref{eq:sys} we have $W\dot{x} = \mathbf{0}$, which implies that every solution $x(t)$ 
with nonnegative initial value $x(0)\in\mathbb{R}_{\ge 0}^s$ remains in the  \defword{stoichiometric compatibility class}
\begin{align}\label{eq:pc}
\mathcal{P}_c := \{x \in \mathbb{R}_{\geq 0}^s \mid Wx = c\},\qquad c := Wx(0) \in \mathbb{R}^d.
\end{align}

For a given rate-constant vector $\kappa^* \in \mathbb{R}_{>0}^m$, a concentration vector $x^* \in \mathbb{R}_{\geq 0}^s$ is called a \defword{steady state} of $G$ if it satisfies
$f(\kappa^*, x^*) = \mathbf{0}$, where $f(\kappa, x)$ denotes the right-hand side of the ODE system~\eqref{eq:sys}. If  all entries of $x^*$ are strictly positive, namely $x^* \in \mathbb{R}_{>0}^s$, then the steady state $x^*$ is called a \defword{positive steady state}.
If there exists $\kappa^* \in \mathbb{R}_{>0}^m$ such that $G$ has a positive steady state, then we say 
$G$ is \defword{consistent}. 
A steady state $x^*$ is said to be \defword{degenerate} if the restriction of the Jacobian matrix $\text{Jac}_f(\kappa^*, x^*)$ to the \defword{stoichiometric subspace} $S$ is not surjective, i.e., \(
\text{im}\left(\text{Jac}_f(\kappa^*, x^*)|_S \right) \neq S.\)
 A network $G$ is called a \defword{degenerate network} if it is consistent and every positive steady state of $G$ is degenerate.  If a network $G$ admits at least one nondegenerate positive steady state, then it is called a \defword{nondegenerate network}. 
Notice that a nondegenerate network must be consistent.   
For a given network \(G\) and a fixed rate-constant vector \(\kappa^*\), 
we say \(G\) has \emph{non-vacuous  ACR} in species \(X_i\) if \(G\) has positive steady states and the value of \(x_i^*\) is identical for every positive steady state \(x^*\). 
We say a network has \emph{non-vacuous ACR} if it has non-vacuous ACR in some species \cite{si2026}.

\subsection{Algebraic criteria for degeneracy}
\label{sec2.2}
For the network $G$ in \eqref{eq:network}, let $\mathcal N$ and $\mathcal X$
be its stoichiometric matrix and reactant matrix, respectively. The \defword{flux cone} of $\mathcal N$ is defined as
\begin{align}\label{eq:Fn}
\mathcal{F}(\mathcal{N}) := \left\{ \gamma \in \mathbb{R}_{\geq 0}^m \mid \mathcal{N} \gamma = \mathbf{0} \right\}.
\end{align}
Let $\ell^{(1)},\ldots,\ell^{(K)}\in\mathbb{R}_{\geq0}^m$ be a set of generators for $\mathcal{F}(\mathcal{N})$. Then, any $\gamma\in \mathcal{F}(\mathcal{N})$ can be written as
\begin{align*}\label{eq:gamma}
\gamma = \sum_{i=1}^{K} \lambda_i \ell^{(i)},~ \text{with } \lambda_i \geq 0 \ \text{for any } i \in \{1, \ldots, K\}.
\end{align*}
We define $\lambda:=(\lambda_1, \ldots, \lambda_K)^\top$.
The \defword{partial transformed Jacobian matrix} in terms of \(\lambda\) is
defined as
\begin{align}
\label{eq:a}
A(\lambda) := \mathcal{N}\operatorname{diag}\left( \gamma \right)\mathcal{X}^{\top}=\mathcal{N}\operatorname{diag}\left( \sum_{i=1}^{K} \lambda_i \ell^{(i)} \right)\mathcal{X}^{\top}.
\end{align}

\begin{lemma}\defword{\cite[Theorem 3.3]{TangWangZhang2026}}
\label{lem2.1}
    Let $G$ be an $r$-dimensional  network with $s$ species. Assume that $A(\lambda)$ is the partial transformed Jacobian matrix defined in \eqref{eq:a}. Then, the network $G$ is degenerate if and only if $\det(A(\lambda)[I,I])$ is the zero polynomial for all $I \subseteq \{1,\dots,s\}$ with $|I| = r$, where
    \( A(\lambda)[I, I] \) denotes the submatrix  with entries of $A(\lambda)$ with indices \( (i, j) \) in \( I \times I \). 
\end{lemma}

\begin{lemma}
\label{lem2.2}
Let $G$ be an $r$-dimensional  network with $s$ species.  Let $\widetilde{G}$ be an $r$-dimensional network obtained from $G$ by removing some species. If $G$ is degenerate, then $\widetilde{G}$ is also degenerate. 
\end{lemma}
\begin{proof}
 Let $A(\lambda)$ be the partial transformed Jacobian matrix of $G$ defined as in \eqref{eq:a}. Since the network $G$ is degenerate, by Lemma~\ref{lem2.1},  it holds that
\begin{align*}\det(A( \lambda)[I,I]) = 0~ \text{for all } I \subseteq \{1,\dots,s\} \text{ with } |I| = r.\end{align*}
Without loss of generality, assume that the species in $\widetilde{G}$ are $X_1,\ldots, X_{\widetilde{s}}$. Note that the network $\widetilde{G}$ remains $r$-dimensional, and its partially transformed Jacobian matrix is exactly the submatrix $A(\lambda)[\{1,\ldots,\widetilde{s}\},\{1,\ldots,\widetilde{s}\}]$. 
Hence,  by Lemma~\ref{lem2.1}, the network $\widetilde{G}$ is also degenerate.
\end{proof}

\section{Main Result}
\label{sec3}
In this section, we establish the main result of the paper. 
Theorem~\ref{thm:main} provides a complete structural classification of 
degenerate two-dimensional zero-one networks in terms of two elementary 
prototypes (lumping networks) and their species refinements. 
Theorem~\ref{thm:deg-char} gives the equivalent matrix-theoretic 
formulation, showing that degeneracy is determined entirely by the row 
patterns of the stoichiometric and reactant matrices and can be verified 
by purely structural inspection of the pair $(\mathcal{N},\mathcal{X})$ 
without computation. 
Remark~\ref{rem:binomial} further reveals that every such degenerate 
network has a steady-state system consisting entirely of binomials, so 
its steady-state variety is toric.

The practical scope of these criteria is demonstrated through three 
examples. 
Example~\ref{ex:redox} applies a simple row-count obstruction derived 
from Theorem~\ref{thm:deg-char} to rule out non-vacuous ACR in a 
coarse-grained redox-energy coupling network. 
Example~\ref{ex:multi} exhibits a three-dimensional network where 
multistationarity coexists with the ACR--degeneracy structure. 
Example~\ref{ex:idhkp} treats a three-dimensional,  biologically derived IDHKP--IDH 
regulatory system and verifies that removing its ACR species produces a 
degenerate subnetwork captured precisely by Theorem~\ref{thm:main}.

\begin{definition}
\label{def:species-refinement}
Let $G$ be a network defined as in \eqref{eq:network} with species set
\(
\mathcal S=\{X_1,\ldots,X_s\}.
\)
For each $i\in\{1,\ldots,s\}$, let
$\phi(X_i):=\sum_{k=1}^{n_i}Z_{i,k}$ $(n_i\geq 0)$,   
where each $Z_{i,k}$ denotes a  species, and $\phi(X_i):=0$ if $n_i=0$. Assume that
\begin{align*}
  \{Z_{i,1},
  \ldots, Z_{i,n_i}\}\cap\{Z_{j,1},\ldots, Z_{j,n_j}\}=\emptyset
  \quad\text{for all distinct}\; i, j\in \{1,\ldots,s\}.  
\end{align*}
Define $\phi(G)$ to be the following network:
\begin{align*}
\label{eq:ref}
\sum_{i=1}^{s}\alpha_{ij}\phi(X_i)
~ \xrightarrow{\kappa_j} ~
\sum_{i=1}^{s}\beta_{ij}\phi(X_i),
 \quad 
    \text{for } j=1, \ldots, m.
\end{align*}
If  $\phi(G)$ has the same dimension as $G$, 
then $\phi(G)$ is called a
\defword{species refinement} of $G$, and $G$ is called a
\defword{lumping network} of $\phi(G)$.
\end{definition}
\begin{theorem}[Main Theorem]\label{thm:main}
A two-dimensional zero-one network without trivial species is degenerate if and only if it is a consistent subnetwork of a species refinement of  Fig.~\ref{fig:12} (\subref{fig:1}) or  Fig.~\ref{fig:12} (\subref{fig:2}).
\end{theorem}
\begin{figure}[H]
    \centering
    \begin{subfigure}{0.495\textwidth}
        \centering
 \includegraphics[width=\textwidth]{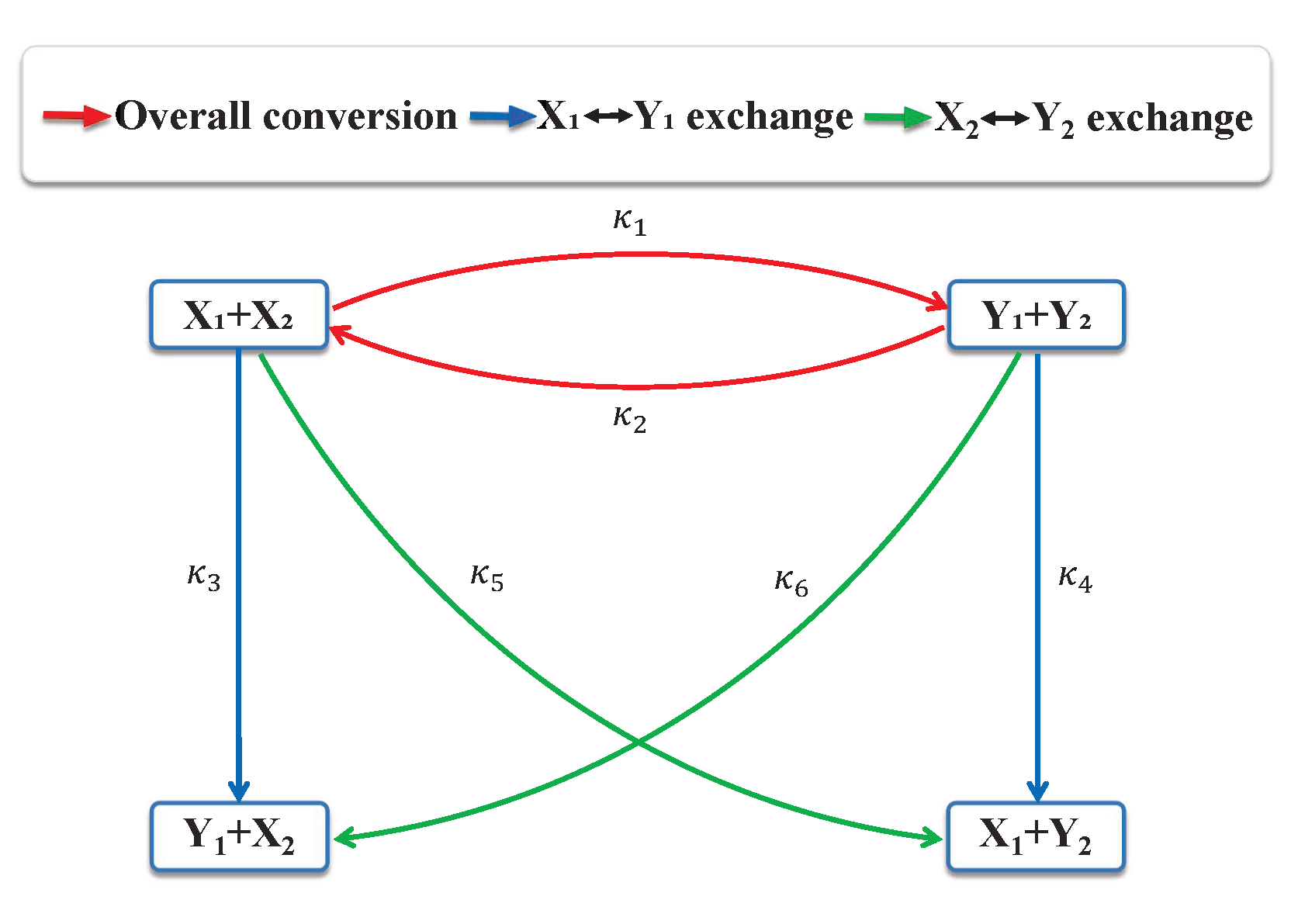}
        \caption{Complete paired-exchange network}
        \label{fig:1}
    \end{subfigure}
    \hfill
    \begin{subfigure}{0.48\textwidth}
        \centering
 \includegraphics[width=\textwidth]{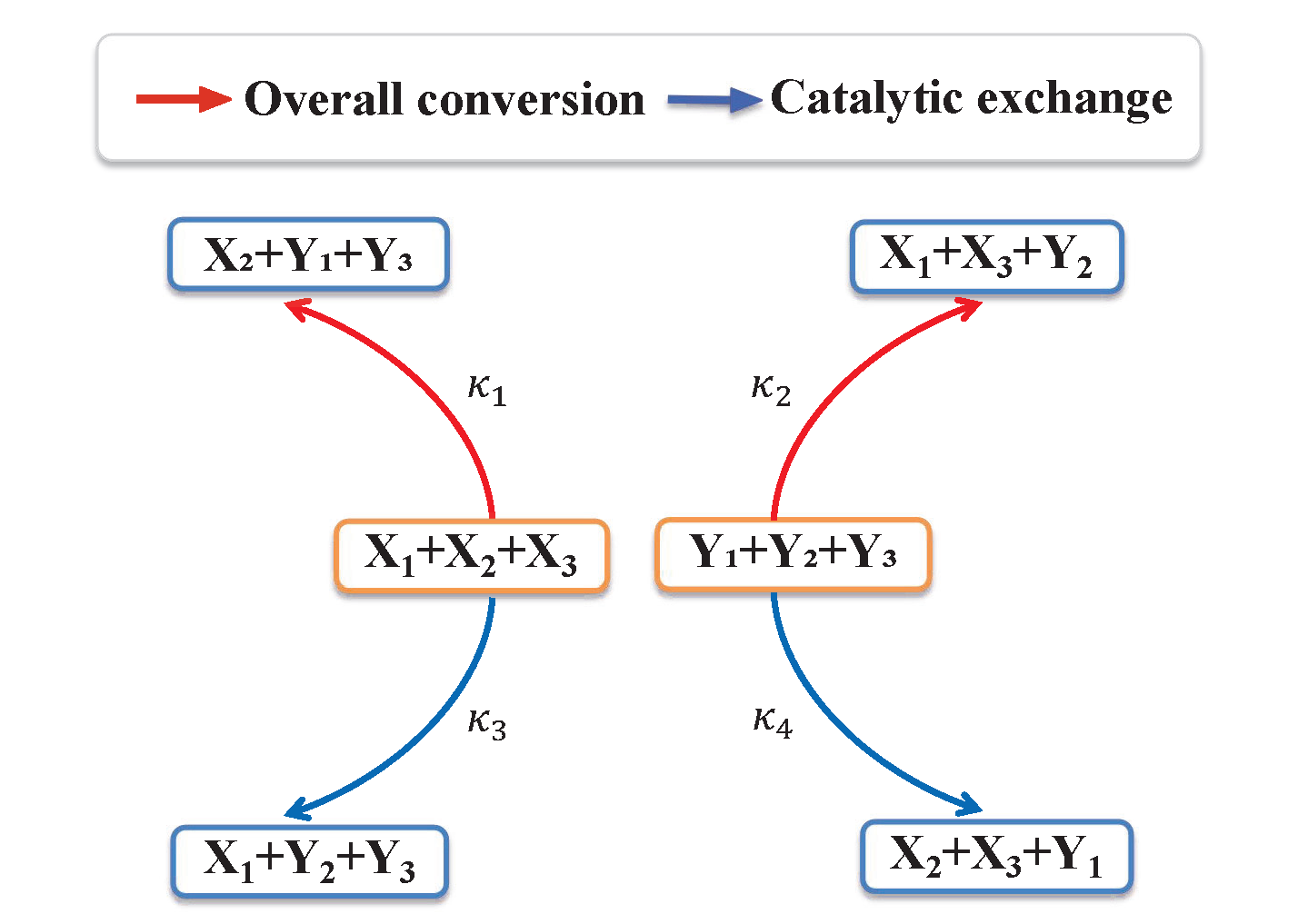}
        \caption{ Catalytic-pair conversion network}
  \label{fig:2}
    \end{subfigure}
    \caption{Two lumping networks stated in Theorem~\ref{thm:main}}
    \label{fig:12}
\end{figure}
\begin{remark}
\label{rem:sub}
Note that the network in Fig.~\ref{fig:12} (\subref{fig:1}) has five consistent subnetworks (see Fig.~\ref{fig:34567}), while that in Fig.~\ref{fig:12} (\subref{fig:2}) has none.
\end{remark}

The following Theorem \ref{thm:deg-char} gives an equivalent matrix-theoretic formulation of Theorem~\ref{thm:main}.

\begin{definition}
\label{def3.42}
Let $\boldsymbol{\alpha}=(\alpha_1,\ldots,\alpha_n)\in\{0,1\}^n$. The \defword{complement} of $\boldsymbol{\alpha}$ is  $\bar{\boldsymbol{\alpha}}:=(\bar\alpha_1,\ldots,\bar\alpha_n)\in\{0,1\}^n$, where $\bar\alpha_i$ is the negative of $\alpha_i$ modulo $2$, i.e. $\bar\alpha_i=1-\alpha_i$ for every $i\in\{1,\ldots,n\}$.
\end{definition}
\begin{theorem}[Matrix Version of Main Theorem]
\label{thm:deg-char}
Let $G$ be a two-dimensional zero-one network with $s$ species and no trivial species. Let $\mathcal{N}$ and $\mathcal{X}$ be the stoichiometric
matrix and the reactant matrix of $G$, respectively.    Then, $G$ is degenerate if and only if one of the following two statements holds. 
\begin{itemize}
\item[(I)]
For every
$i\in\{1,\ldots,s\}$, we have
\begin{align}\label{eq:b.7}
(\mathcal{N}_i,\mathcal{X}_i)
\in
\bigl\{
(\mathcal{E}_1,\mathcal{Y}),\,
(-\mathcal{E}_1,\bar{\mathcal{Y}}),\,
(\mathcal{E}_2,\mathcal{Y}),\,
(-\mathcal{E}_2,\bar{\mathcal{Y}})
\bigr\},
\end{align}
where 
$\begin{pmatrix}
\mathcal{E}_1\\
\mathcal{E}_2
\end{pmatrix}
$  and $\begin{pmatrix}
\mathcal{Y}\\
\mathcal{Y}
\end{pmatrix}
$  
correspond to a consistent subnetwork of 
\begin{align}
\mathcal{N}=\left(
\begin{array}{cccccc}
-1 & 1 & -1 & 1 & 0 & 0\\
-1 & 1 & 0 & 0 &-1 & 1
\end{array}\right), ~\mathcal{X}=\left(
\begin{array}{cccccc}
1 & 0 & 1 & 0 & 1 & 0\\
1 & 0 & 1 & 0 & 1 & 0
\end{array}\right).\label{eq:mainnet1}
\end{align}
\item[(II)] 
For every
$i\in\{1,\ldots,s\}$, we have
\begin{align}\label{eq:b.1}
(\mathcal{N}_i,\mathcal{X}_i)
\in
\bigl\{
(\mathcal{E}_1,\mathcal{Y}),\,
(-\mathcal{E}_1,\bar{\mathcal{Y}}),\,
(\mathcal{E}_2,\mathcal{Y}),\,
(-\mathcal{E}_2,\bar{\mathcal{Y}}),\,
(\mathcal{E}_3,\mathcal{Y}),\,
(-\mathcal{E}_3,\bar{\mathcal{Y}})
\bigr\},
\end{align}
where 
$\begin{pmatrix}
\mathcal{E}_1\\
\mathcal{E}_2 \\
\mathcal{E}_3
\end{pmatrix}
$  and $\begin{pmatrix}
\mathcal{Y}\\
\mathcal{Y}\\
\mathcal{Y}
\end{pmatrix}
$ 
correspond to a consistent subnetwork of 
\begin{align}
\label{eq:mainnet2}
{\mathcal N}=\left(
\begin{array}{cccc}
-1 & 1 & 0 & 0 \\
0 & 0 & -1 & 1 \\
-1 & 1 & -1 & 1 
\end{array}\right),~\mathcal{X}=\left(
\begin{array}{cccc}
1 & 0 & 1 & 0 \\
1 & 0 & 1 & 0 \\
1 & 0 & 1 & 0 
\end{array}\right). 
\end{align}
\end{itemize}
\end{theorem}

\begin{figure}[htbp]
    \centering
    \begin{subfigure}{0.48\textwidth}
        \centering
        \includegraphics[width=\textwidth]{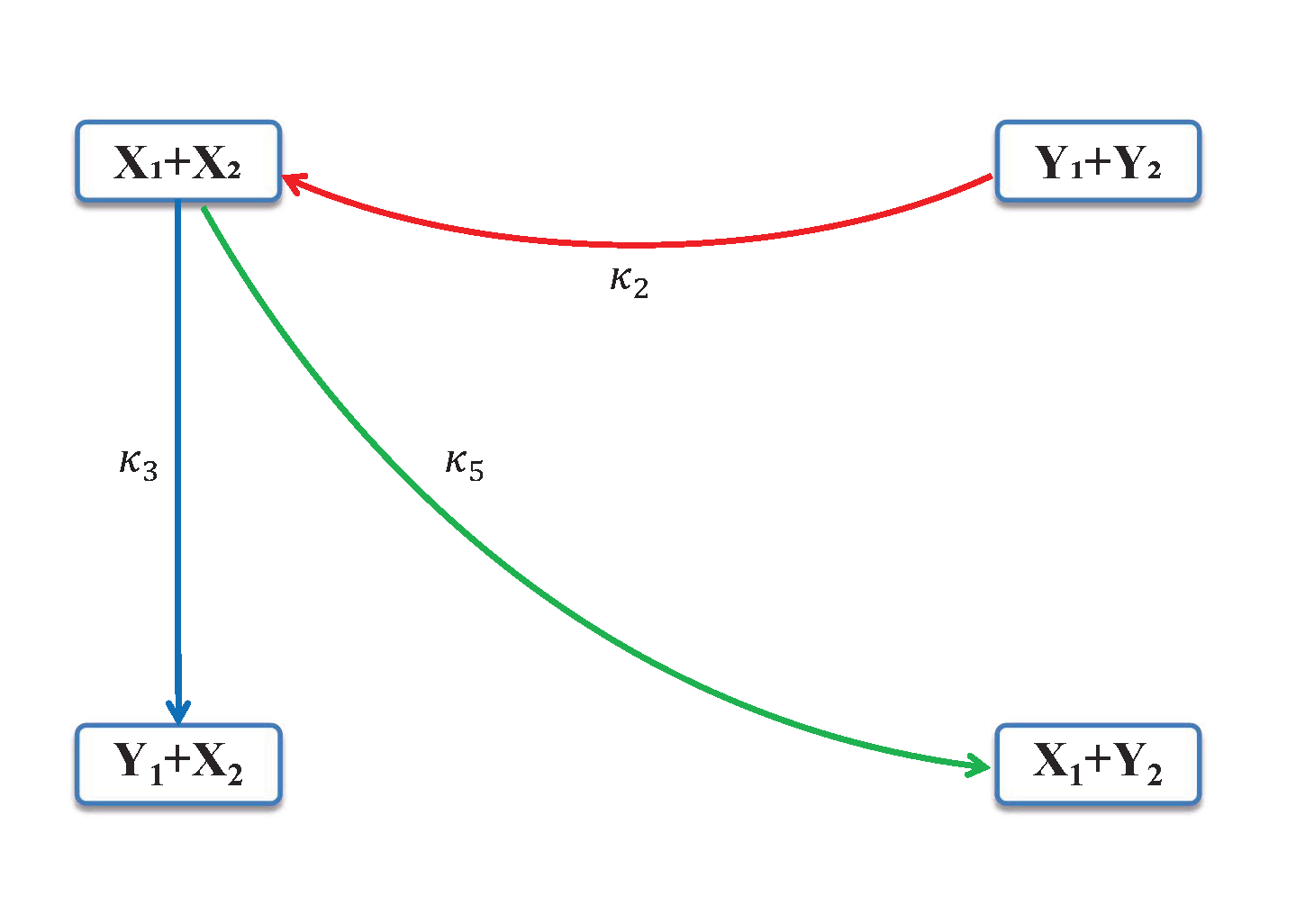}
        \caption{}
        \label{fig:3}
    \end{subfigure}
    \hfill
    \begin{subfigure}{0.48\textwidth}
        \centering
        \includegraphics[width=\textwidth]{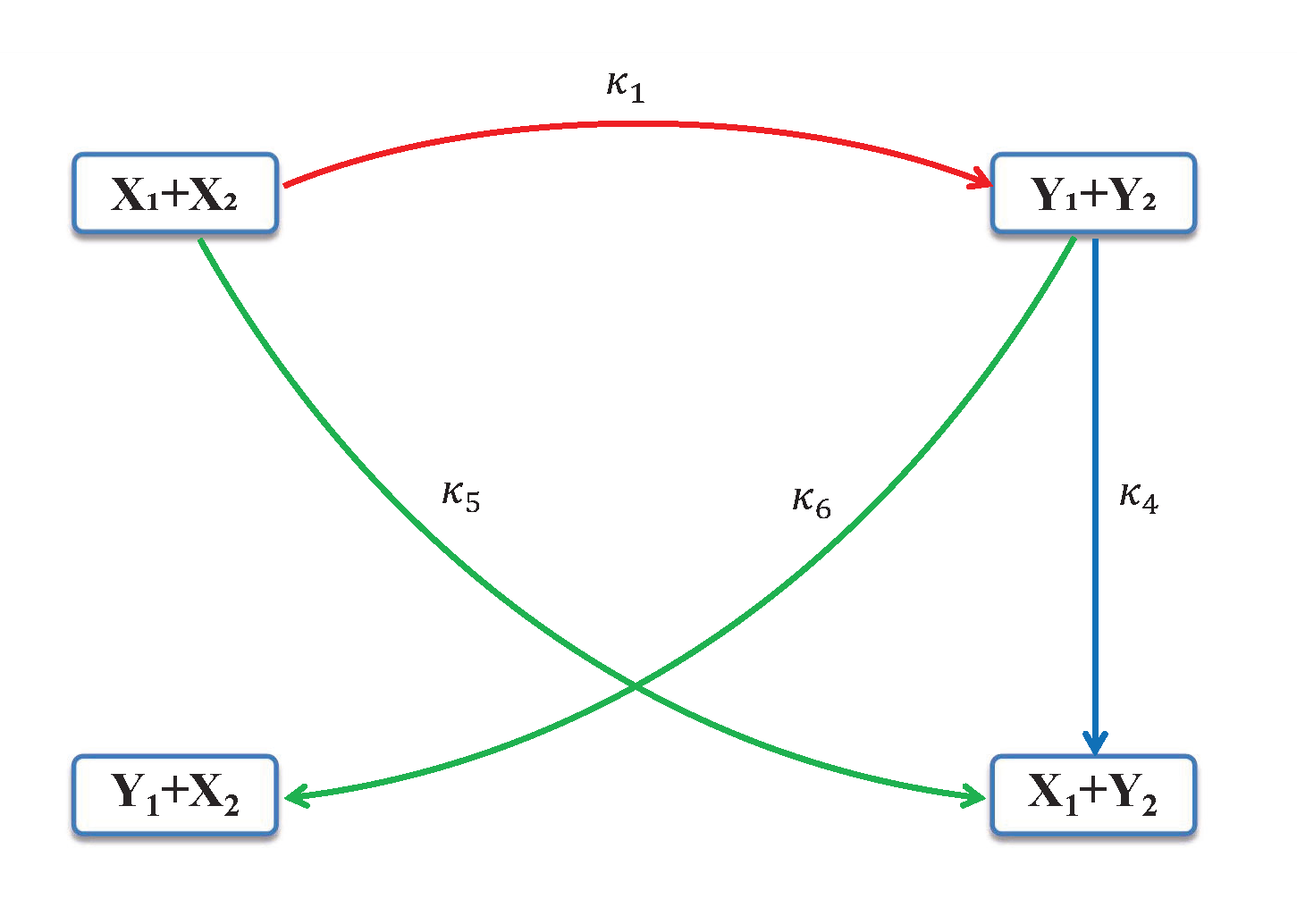}
        \caption{}
        \label{fig:4}
    \end{subfigure}
    \begin{subfigure}{0.48\textwidth}
        \centering
        \includegraphics[width=\textwidth]{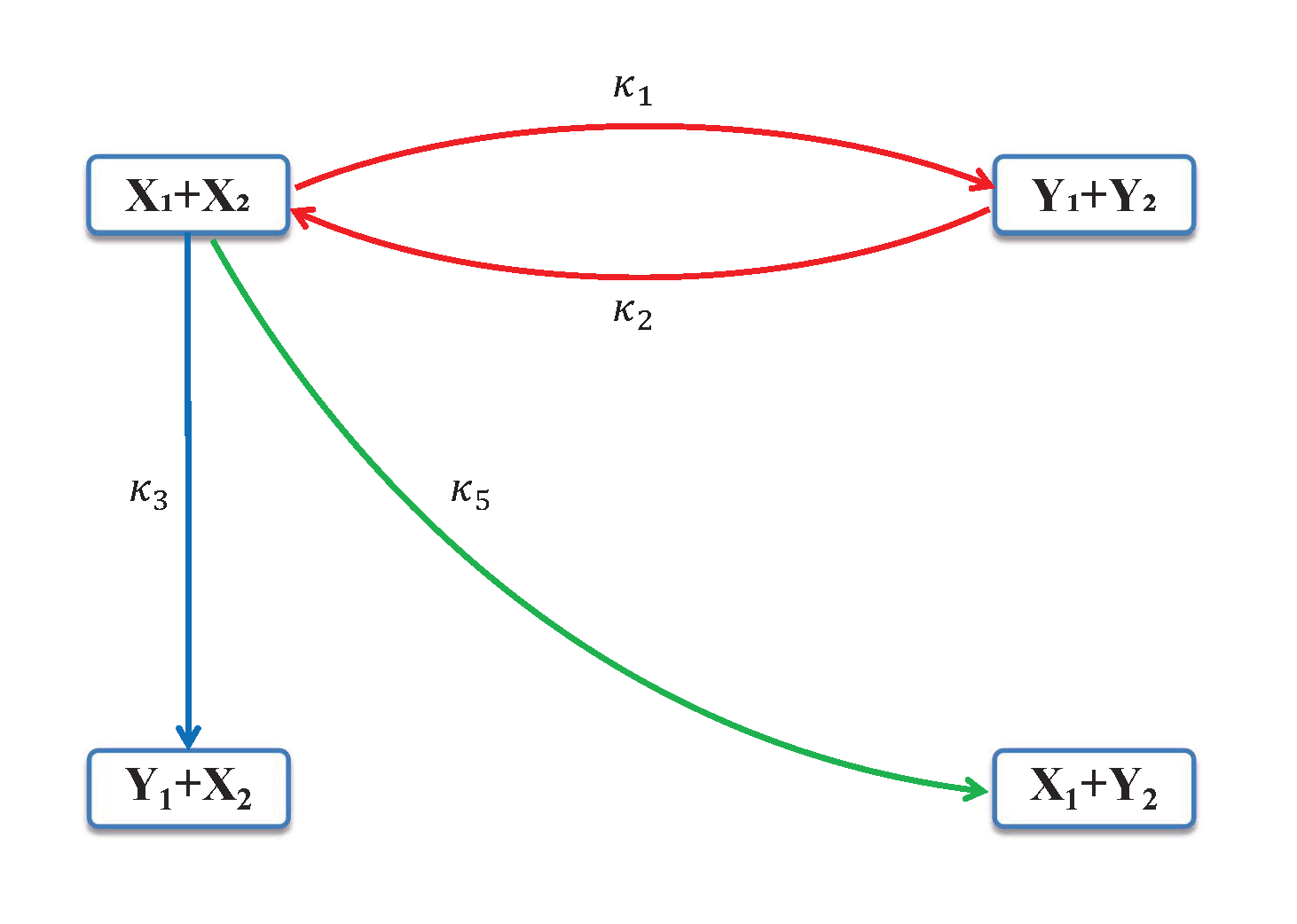}
        \caption{}
        \label{fig:5}
    \end{subfigure}
    \hfill
    \begin{subfigure}{0.48\textwidth}
        \centering
        \includegraphics[width=\textwidth]{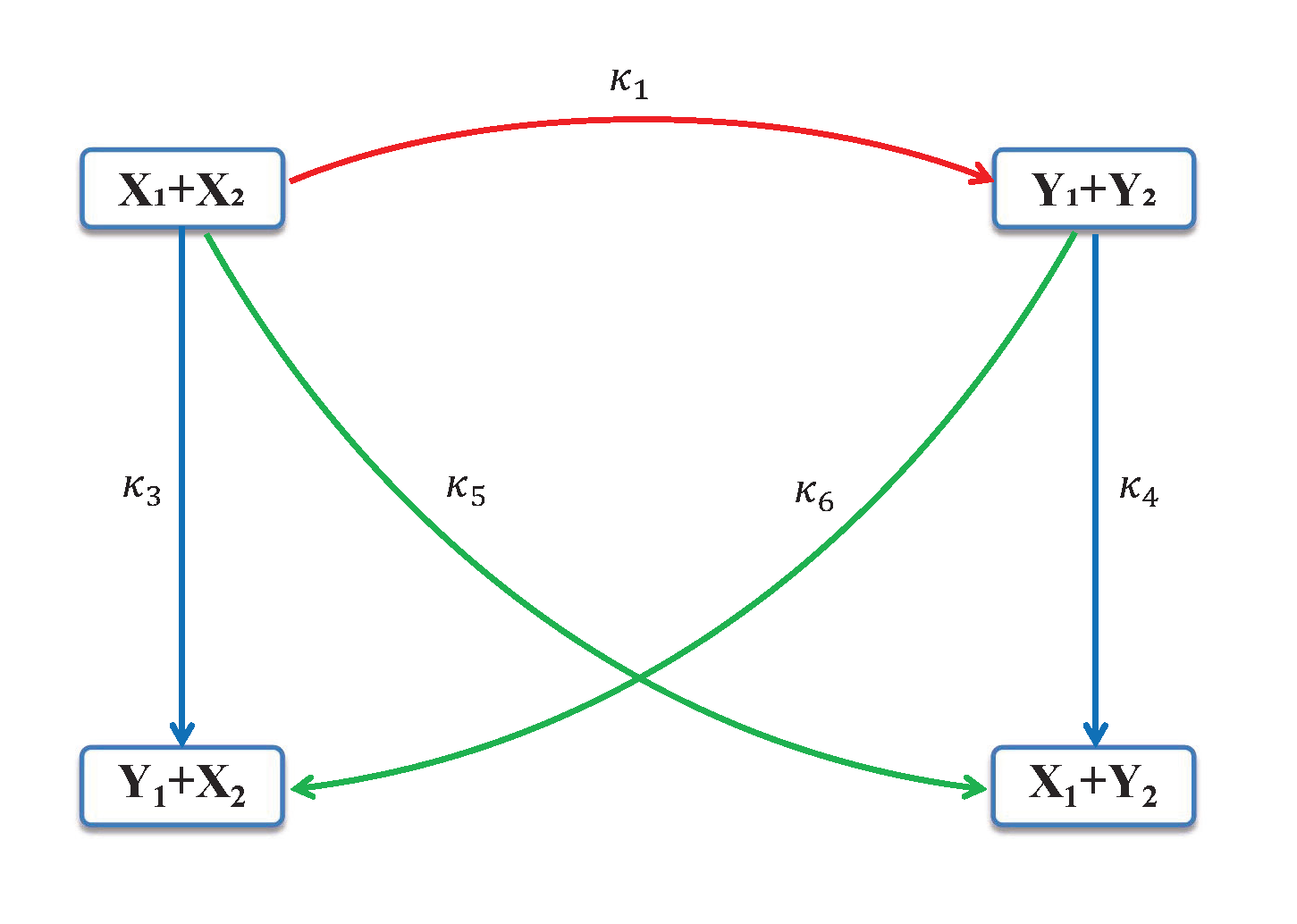}
        \caption{}
        \label{fig:6}
    \end{subfigure}
    \begin{subfigure}{0.48\textwidth}
        \centering
        \includegraphics[width=\textwidth]{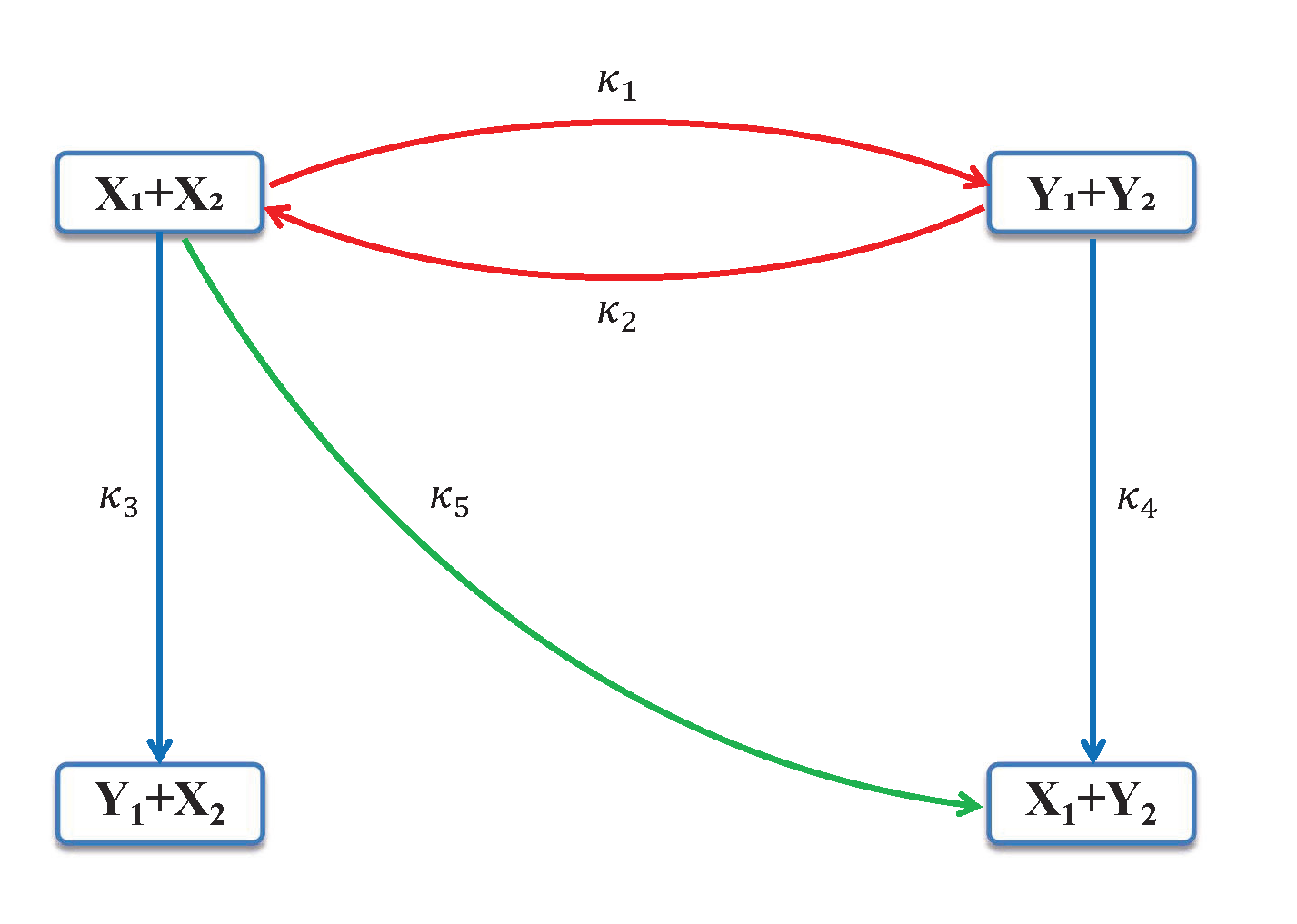}
        \caption{}
        \label{fig:7}
    \end{subfigure}
    \caption{Consistent subnetworks of Fig.~\ref{fig:12} (\subref{fig:1})}
    \label{fig:34567}
\end{figure}

\begin{remark}
\label{rem:binomial}
For the network in Fig.~\ref{fig:12} (\subref{fig:1}), the steady-state system is
\begin{align*}
\left\{
\begin{aligned}
f_1&=-(\kappa_1+\kappa_3)x_1x_2
     +(\kappa_2+\kappa_4)y_1y_2,\\
f_2&=-(\kappa_1+\kappa_5)x_1x_2
     +(\kappa_2+\kappa_6)y_1y_2,\\
f_3&=-f_1,~f_4=-f_2.
\end{aligned}\right.
\end{align*}
Under any species refinement, replacing $x_i$ and $y_i$ by
$\prod_{k=1}^{p_i}x_{i,k}$ and
$\prod_{k=1}^{q_i}y_{i,k}$, respectively, turns the first two polynomials into
\begin{align*}
\left\{
\begin{aligned}
f_1&=-(\kappa_1+\kappa_3)
\prod_{i=1}^{2}\prod_{k=1}^{p_i}x_{i,k}
+(\kappa_2+\kappa_4)
\prod_{i=1}^{2}\prod_{k=1}^{q_i}y_{i,k},\\
f_2&=-(\kappa_1+\kappa_5)
\prod_{i=1}^{2}\prod_{k=1}^{p_i}x_{i,k}
+(\kappa_2+\kappa_6)
\prod_{i=1}^{2}\prod_{k=1}^{q_i}y_{i,k}.
\end{aligned}
\right.
\end{align*}
Similarly, for the network in Fig.~\ref{fig:12} (\subref{fig:2}), the steady-state system is
\begin{align*}
\left\{
\begin{aligned}
&f_1=-\kappa_1x_1x_2x_3+\kappa_2y_1y_2y_3,
~
f_2=-\kappa_3x_1x_2x_3+\kappa_4y_1y_2y_3,\\
&f_3=f_1+f_2,
~f_4=-f_1,
~f_5=-f_2,
~f_6=-f_1-f_2.
\end{aligned}
\right.
\end{align*}
Under any species refinement, the first two polynomials become
\begin{align*}
\left\{
\begin{aligned}
f_1
&=-\kappa_1
\prod_{i=1}^{3}\prod_{k=1}^{n_i}x_{i,k}
+\kappa_2
\prod_{i=1}^{3}\prod_{k=1}^{s_i}y_{i,k},\\
f_2
&=-\kappa_3
\prod_{i=1}^{3}\prod_{k=1}^{n_i}x_{i,k}
+\kappa_4
\prod_{i=1}^{3}\prod_{k=1}^{s_i}y_{i,k}.
\end{aligned}
\right.
\end{align*}
In both cases, every $f_i$ is a linear combination of $f_1$ and $f_2$, which share the same two monomials. Consequently, for both networks and any consistent subnetwork of their species refinements, every steady-state polynomial is a binomial. By Theorem~\ref{thm:main}, the same holds for every degenerate two-dimensional zero-one network without trivial species. Thus, Theorem~\ref{thm:main} generalizes \cite[Theorem~6.1]{TangWangZhang2026}, which was obtained purely by computational enumeration.
\end{remark}

\begin{example}
\label{ex:redox}
Consider the following coarse-grained biochemical network \(G\) consisting of
two reversible reactions coupled through the redox pair
NADH/NAD$^+$ and the energy pair ATP/ADP:
\begin{align}
\label{eq:exam1}
S_1+\mathrm{NADH}+\mathrm{ATP}
&\xrightleftharpoons[\kappa_2]{\kappa_1}
P_1+\mathrm{NAD}^++\mathrm{ADP}, \notag\\
S_2+\mathrm{NADH}+\mathrm{ADP}
&\xrightleftharpoons[\kappa_4]{\kappa_3}
P_2+\mathrm{NAD}^++\mathrm{ATP}.
\end{align}
For convenience, denote the $8$ species by
\begin{align*}
(S_1,P_1,S_2,P_2,\mathrm{NADH},\mathrm{NAD}^+,
\mathrm{ATP},\mathrm{ADP})
=
(X_1,Y_1,X_2,Y_2,X_3,Y_3,X_4,Y_4).   
\end{align*}
With the species ordered as
\(
X_1,Y_1,X_2,Y_2,X_3,Y_3,X_4,Y_4,
\)
the stoichiometric matrix is
\begin{align*}
 \mathcal N=
\begin{pmatrix}
-1& 1& 0& 0\\
 1&-1& 0& 0\\
 0& 0&-1& 1\\
 0& 0& 1&-1\\
-1& 1&-1& 1\\
 1&-1& 1&-1\\
-1& 1& 1&-1\\
 1&-1&-1& 1
\end{pmatrix}.   
\end{align*}
Clearly,
\(
\operatorname{rank}(\mathcal N)=2,
\)
and the $8$ rows of $\mathcal N$ are pairwise distinct. By
\cite[Theorem 3.1]{si2026}, if the network has non-vacuous ACR in a species $X_i$ (or $Y_i$) for any
generic choice of rate constants, then
after removing the species, the resulting $7$-species network must be degenerate. However, Theorem \ref{thm:deg-char} implies that any degenerate two-dimensional zero-one network has at most $6$ distinct rows in its stoichiometric matrix. So, we conclude directly that the network  \eqref{eq:exam1} has no non-vacuous ACR for any
generic choice of rate constants.  

On the other hand, we can verify the above conclusion by computation.  
Let $x_i$ and $y_i$ denote the concentrations of $X_i$ and
$Y_i$, respectively. At every positive steady state, we have \(
\kappa_1x_1x_3x_4=\kappa_2y_1y_3y_4~\text{and}~
\kappa_3x_2x_3y_4=\kappa_4y_2y_3x_4.\)
Hence,
\begin{align*}
 y_1=
\frac{\kappa_1}{\kappa_2}
\frac{x_1x_3x_4}{y_3y_4}~\text{and}~
y_2=
\frac{\kappa_3}{\kappa_4}
\frac{x_2x_3y_4}{y_3x_4}.   
\end{align*}
Thus, for any positive rate constants, the positive steady states form a family in which
\(x_1,x_2,x_3,y_3,x_4,y_4\) may vary. Then, no species has a constant
concentration over all positive steady states, and the network \(G\) in \eqref{eq:exam1} has no non-vacuous ACR.
\end{example}

\begin{example}
\label{ex:multi}
Consider the following zero-one network $G$:
\begin{equation*}
\begin{array}{lll}
Z \xrightarrow{\kappa_1} 0,
&
Y_1+Y_2 \xrightarrow{\kappa_2} X_1+X_2,
&
X_1+X_2+Z \xrightarrow{\kappa_3} Y_1+X_2+Z,
\\[3pt]
X_1 \xrightarrow{\kappa_4} X_1+Z,
&
X_1+X_2 \xrightarrow{\kappa_5} X_1+Y_2,
&
X_2 \xrightarrow{\kappa_6} X_2+Z.
\end{array}
\end{equation*}
With the species ordered as
\(
X_1,X_2,Y_1,Y_2,Z,
\)
the stoichiometric matrix is
\begin{align*}
\mathcal N
=
\begin{pmatrix}
 0& 1&-1&0& 0&0\\
 0& 1& 0&0&-1&0\\
 0&-1& 1&0& 0&0\\
 0&-1& 0&0& 1&0\\
-1& 0& 0&1& 0&1
\end{pmatrix}.
\end{align*}
Notice that the network $G$ is
three-dimensional. Its steady-state system is
\begin{align*}
\left\{
\begin{aligned}
f_1&=\kappa_2y_1y_2-\kappa_3x_1x_2z,
~
f_2=\kappa_2y_1y_2-\kappa_5x_1x_2,\\
f_3&=-f_1,~
f_4=-f_2,~
f_5=-\kappa_1z+\kappa_4x_1+\kappa_6x_2.
\end{aligned}
\right.
\end{align*}
At any positive steady state, \(f_1=f_2=0\) implies \(
\kappa_3x_1x_2z=\kappa_5x_1x_2,
\)
and hence \(
z=\kappa_5/\kappa_3.
\)
Thus, the network $G$ has ACR in species \(Z\). 

Now remove the ACR species \(Z\). After deleting the resulting trivial
reactions, the resulting network is
\begin{align*}
Y_1+Y_2 \xrightarrow{\kappa_2}X_1+X_2,~X_1+X_2 \xrightarrow{\kappa_3} Y_1+X_2,~X_1+X_2 \xrightarrow{\kappa_5} X_1+Y_2,
\end{align*}
which is exactly the network in Fig.~\ref{fig:34567} (\subref{fig:3}).  This example shows that a two-dimensional degenerate network characterized by Theorem~\ref{thm:main} can arise naturally by removing an ACR species from a higher-dimensional network. Conversely, these classified degenerate networks can be used
as building blocks for constructing higher-dimensional networks with
ACR, thereby providing a connection between network degeneracy and ACR. Notably, despite possessing ACR in species $Z$, the network $G$ admits at most two positive steady states in any stoichiometric compatibility class $\mathcal{P}_c$ \eqref{eq:pc}, demonstrating that multistationarity can coexist with the ACR-degeneracy structure.
\end{example}

\begin{example}
\label{ex:idhkp} 
Consider the following modified network \(G\) based on the IDHKP--IDH glyoxylate bypass regulation system \cite[Example~3.3]{joshi2024}:
\begin{align}
\label{eq:exam3}
&\text{\textit{Overall conversion}:} && E + EI_pI + I \xrightarrow{\kappa_1} I_p + EI_p, \notag\\
&\text{\textit{Phosphorylation}:} && E + EI_pI \xrightarrow{\kappa_2} I_p + EI_pI, \notag\\
&\text{\textit{Complex formation}:} && I_p + EI_p + I \xrightarrow{\kappa_3} I_p + EI_pI, \notag\\
&\text{\textit{Complex dissociation}:} && E + EI_pI \xrightarrow{\kappa_4} E + EI_p + I, \notag\\
&\text{\textit{Dephosphorylation}:} && I_p + EI_p \xrightarrow{\kappa_5} E + EI_p + I,
\end{align}
where $I$ and $I_p$ denote the active unphosphorylated and phosphorylated forms of isocitrate dehydrogenase (IDH), respectively, while $E$ denotes the bifunctional regulatory enzyme IDHKP, and $EI_p$, $EI_pI$ are intermediate enzyme--substrate complexes. The network retains the main motifs of complex formation (third reaction) and dissociation (fourth reaction), together with the coupling between the phosphorylated and unphosphorylated forms of IDH through the phosphorylation cycle (second and fifth reactions) and the overall conversion (first reaction).

With the species ordered as $E, EI_pI, I_p, EI_p, I$, the stoichiometric matrix is
\begin{align*}
\mathcal{N}=
\begin{pmatrix}
-1&-1& 0& 0& 1\\
-1& 0& 1&-1& 0\\
 1& 1& 0& 0&-1\\
 1& 0&-1& 1& 0\\
-1& 0&-1& 1& 1
\end{pmatrix},
\end{align*}
so $\operatorname{rank}(\mathcal{N})=3$ and the network is three-dimensional. 
Let $x_1,x_2,y_1,y_2,z$ denote the concentrations of 
$E,EI_pI,I_p,EI_p,I$, respectively. From the steady-state equations,
\[
\kappa_5 y_1 y_2=\kappa_1 x_1 x_2 z+\kappa_2 x_1 x_2=2\kappa_1 x_1 x_2 z.
\]
Since $x_1x_2>0$ at any positive steady state, 
we obtain $z=\kappa_2/\kappa_1$. Hence, the network $G$ in \eqref{eq:exam3} has ACR in species $I$.

Removing the ACR species $I$ and relabeling 
$E,EI_pI,I_p,EI_p$ as $X_1,X_2,Y_1,Y_2$, respectively, yields the 
subnetwork
\begin{align*}
X_1+X_2 &\xrightarrow{\kappa_1} Y_1+Y_2, &
X_1+X_2 &\xrightarrow{\kappa_2} Y_1+X_2, &
Y_1+Y_2 &\xrightarrow{\kappa_3} Y_1+X_2,\\
X_1+X_2 &\xrightarrow{\kappa_4} X_1+Y_2, &
Y_1+Y_2 &\xrightarrow{\kappa_5} X_1+Y_2. &
\end{align*}
After an appropriate relabeling of the rate constants, this is precisely 
the network shown in Fig.~\ref{fig:34567} (\subref{fig:6}).
\end{example}


\section{Proof}
\label{sec4}

In this section, we prove Theorem~\ref{thm:deg-char}, which shows that a two-dimensional zero-one network without trivial species is degenerate if and only if its stoichiometric and reactant matrices satisfy the structure in
\eqref{eq:b.7} or \eqref{eq:b.1}. Theorem~\ref{thm:main} then follows from this characterization. This section is organized as follows.

In Section~\ref{sec4.1}, all degenerate two-species networks are characterized, as stated in Theorem~\ref{thm:2s}. To prove this result, we first recall in Lemma~\ref{lem:lift} the classical
result that nondegeneracy can be lifted from a subnetwork to the whole network. The  column vectors of the stoichiometric matrix are then divided into the four sets \(B_1,B_2,B_3,\) and \(B_4\). Lemma~\ref{lm:b1b2} rules out the case where columns from both
$B_1$ and $B_2$ occur. Based on this result, Lemma~\ref{lm:nece}
restricts the possible forms of the stoichiometric matrix. The remaining cases are then characterized in Lemma~\ref{lm:2cases}, completing the proof of Theorem~\ref{thm:2s}.

In Section~\ref{sec4.2}, we complete the proof of Theorem~\ref{thm:deg-char}. 
Lemma~\ref{lm:cl} gives restrictions on the coefficients in the conservation laws. 
Combining this lemma with the two-species characterization in Theorem~\ref{thm:2s} yields \eqref{eq:b.7} and \eqref{eq:b.1}, and completes the proof of Theorem~\ref{thm:deg-char}.

\subsection{Networks with two species}
\label{sec4.1}

\begin{theorem}
\label{thm:2s}
Let $G$ be a two-dimensional zero-one network with two species.  Then, the network $G$ is degenerate if and only if it is a consistent subnetwork of one of the following two networks:
\begin{align}
\mathcal{N}=\left(
\begin{array}{cccccc}
-1 & 1 & -1 & 1 & 0 & 0\\
-1 & 1 & 0 & 0 &-1 & 1
\end{array}\right), ~\mathcal{X}=\left(
\begin{array}{cccccc}
1 & 0 & 1 & 0 & 1 & 0\\
1 & 0 & 1 & 0 & 1 & 0
\end{array}\right);\label{eq:2snet1}
\end{align}
\begin{align}
\mathcal{N}=\left(
\begin{array}{cccccc}
-1 & 1 & -1 & 1 & 0 & 0\\
1 & -1 & 0 & 0 & 1 & -1
\end{array}\right), ~\mathcal{X}=\left(
\begin{array}{cccccc}
1 & 0 & 1 & 0 & 1 & 0\\
0 & 1 & 0 & 1 & 0 & 1
\end{array}\right).\label{eq:2snet2}
\end{align}
\end{theorem}

To prove Theorem~\ref{thm:2s}, we establish the following series of lemmas. Theorem~\ref{thm:2s} is proved at the end of this subsection.

\begin{lemma}
\cite[Lemma 4.3]{JoshiShiu2013}
\label{lem:lift}
 If a subnetwork of  $G$ is nondegenerate, then the network $G$ is also nondegenerate.
\end{lemma}

Note that for any two-dimensional zero-one network with two species,  the stoichiometric matrix
\(\mathcal{N}\) must be formed by column vectors from the following sets:
 \begin{align}\label{eq:defb1b2}
    B_1:=\Biggl\{\left(
\begin{array}{c}
-1  \\
-1
\end{array}\right),\left(
\begin{array}{c}
1  \\
1
\end{array}\right)\Biggr\},~B_2:=\Biggl\{\left(
\begin{array}{c}
-1  \\
1
\end{array}\right),\left(
\begin{array}{c}
1  \\
-1
\end{array}\right)\Biggr\},
\end{align}
\begin{align}\label{eq:defb3b4}
    B_3:=\Biggl\{\left(
\begin{array}{c}
-1  \\
0
\end{array}\right),\left(
\begin{array}{c}
1  \\
0
\end{array}\right)\Biggr\},~ B_4:=\Biggl\{\left(
\begin{array}{c}
0  \\
-1
\end{array}\right),\left(
\begin{array}{c}
0  \\
1
\end{array}\right)\Biggr\}.\end{align}

\begin{lemma}
\label{lm:b1b2}
Let $G$ be a two-dimensional zero-one network with two species.  Suppose  \( G \) is consistent. If the stoichiometric matrix $\mathcal{N}$ of  $G$ contains column vectors from both $B_1$ and $B_2$, 
then  $G$ is nondegenerate. 
\end{lemma}

\begin{proof}
We prove the conclusion by considering two cases.

{\it Case I.} 
If \(\mathcal{N}\) is formed only by column vectors from $B_1$ and $B_2$, then  
$G$ is a consistent subnetwork of  
 \begin{align}\label{eq:b1b2}
\mathcal{N}=\left(
\begin{array}{cccc}
-1 & 1 & -1 & 1 \\
-1 & 1 & 1 & -1  
\end{array}\right),\;
\mathcal{X}=\left(
\begin{array}{cccc}
1 & 0 & 1 & 0 \\
1 & 0 & 0 & 1  
\end{array}\right).
\end{align} 
Since $G$ is two-dimensional and consistent, ${\mathcal N}$ has at least $3$ columns. However, it is straightforward to check that any $3$-reaction subnetwork of \eqref{eq:b1b2} is inconsistent. So, the network $G$ must be \eqref{eq:b1b2} itself. 
 Notice that the flux cone $\mathcal{F}(\mathcal{N})$ defined in \eqref{eq:Fn}  has the generators $\ell^{(1)} = (1,1,0,0)^{\top}$ and $\ell^{(2)} = (0,0,1,1)^{\top}$. Thus, for any $\gamma\in \mathcal{F}(\mathcal{N})$, we can write  $\gamma=\lambda_1\ell^{(1)}+\lambda_2\ell^{(2)}=(\lambda_1,\lambda_1,\lambda_2,\lambda_2)^\top$, where $\lambda_1,\lambda_2\geq 0$. By the definition of $A(\lambda)$ in \eqref{eq:a}, we have
\begin{align*}
   \det(A(\lambda)) &= \det\left(\mathcal{N}\operatorname{diag}(\gamma)\mathcal{X}^\top \right) \notag \\ 
&=4\lambda_1\lambda_2\not\equiv 0.
\end{align*}
 By Lemma~\ref{lem2.1}, the network \( G \) is nondegenerate.

 {\it Case II.} 
If \(\mathcal{N}\) contains at least one column vector from $B_1$, $B_2$ and $B_3\cup B_4$ respectively, then there exists a submatrix of
$\mathcal{N}$ that is either
\begin{align}\label{eq:b1b2b4net}
\begin{pmatrix}
-1 & 1 & 0 \\
-1 & -1 & 1
\end{pmatrix}\;\text{or}\;
\begin{pmatrix}
-1 & 1 & 0 \\
1 & 1 & -1
\end{pmatrix}
\end{align}
up to a permutation of rows or columns. 
For the first matrix in \eqref{eq:b1b2b4net},  the  corresponding   reactant matrix  can be one of the following two matrices:
\[
\begin{array}{cc}
\begin{pmatrix}
1 & 0 & 0 \\
1 & 1 & 0
\end{pmatrix},
&
\begin{pmatrix}
1 & 0 & 1 \\
1 & 1 & 0
\end{pmatrix}.
\end{array}
\]
In both cases, we can verify by Lemma~\ref{lem2.1} that the corresponding subnetwork is nondegenerate. 
Similarly, for the second matrix in \eqref{eq:b1b2b4net}, we have the same conclusion. So, by Lemma \ref{lem:lift}, the network $G$ is nondegenerate.
\end{proof}

\begin{lemma}
\label{lm:nece}
  Let $G$ be a two-dimensional zero-one network with two species.  If $G$ is degenerate, then $\mathcal{N}$ is a submatrix of one of the following two matrices:
\begin{align}
\left(
\begin{array}{cccccccccc}
-1 & 1&  -1 & -1 & 1 & 1 & 0 & 0 & 0 & 0\\
-1 & 1& 0  & 0 & 0 & 0 & -1 & -1 & 1 & 1 
\end{array}\right),
\label{eq:n2}
\end{align}
\begin{align}
\left(
\begin{array}{cccccccccc}
-1 & 1&  -1 & -1 & 1 & 1 & 0 & 0 & 0 & 0\\
1 & -1& 0  & 0 & 0 & 0 & -1 & -1 & 1 & 1 
\end{array}\right)
\label{eq:n3}
\end{align}
up to a permutation of rows or columns. 
\end{lemma}

\begin{proof}
By Lemma \ref{lm:b1b2}, 
if column vectors from \(B_1\) are selected but none from \(B_2\), while vectors from sets \(B_3\) and \(B_4\) can be selected, then \(\mathcal{N}\) is a submatrix of \eqref{eq:n2}. If column vectors from \(B_1\) are not selected but those from \(B_2\) are, and vectors from sets \(B_3\) and \(B_4\) can again be selected, then \(\mathcal{N}\) is a submatrix of \eqref{eq:n3}. 
\end{proof}

\begin{lemma}
\label{lm:2cases}
Let $G$ be a two-dimensional zero-one network with two species.
\begin{itemize}
\item[(I)] If the stoichiometric matrix \(\mathcal{N}\) is a submatrix of \eqref{eq:n2},  then the network \(G\) is degenerate if and only if it is a consistent subnetwork of \eqref{eq:2snet1}.
\item[(II)] If the stoichiometric matrix \(\mathcal{N}\) is a submatrix of \eqref{eq:n3},  then the network \(G\) is degenerate if and only if it is a consistent subnetwork of \eqref{eq:2snet2}.
\end{itemize}
\end{lemma}

\begin{proof}

{\it (I)} $\Longleftarrow$ It is straightforward to check by Lemma \ref{lem2.1} that the network \eqref{eq:2snet1} is degenerate. So, by Lemma \ref{lem:lift}, any consistent subnetwork of \eqref{eq:2snet1} is also degenerate. 

$\Longrightarrow$ 
Notice that the network \eqref{eq:2snet1} is a $6$-reaction subnetwork of some network corresponding to the matrix \eqref{eq:n2}. By Lemma \ref{lm:nece}, we only need to show that any other submatrix of  \eqref{eq:n2} with $6$ columns and any submatrix of  \eqref{eq:n2} with more than $6$ columns must correspond to nondegenerate networks. Notice that any $6$-column submatrix of \eqref{eq:n2} different from
$\mathcal N$ stated in \eqref{eq:2snet1}, or any $m$-column submatrix
($m>6$) of \eqref{eq:n2}, must contain a $5$-column submatrix that is
either $E_1$ or $E_2$, where
\begin{align*}
E_1=\left(
\begin{array}{ccccc}
-1 & 1 & 0 & 0 & 0\\
0 & 0 & -1 & 1 & 1 
\end{array}\right)~\mbox{or}~E_2=\left(
\begin{array}{ccccc}
-1 & 1 & 0 & 0 & 0\\
0 & 0 & 1 & -1 & -1 
\end{array}\right).
\label{eq:l11}
\end{align*}
Next, we show that any network corresponding to $E_1$ or $E_2$ is
nondegenerate. The conclusion then follows directly from
Lemma~\ref{lem:lift}. Notice that the flux cone $\mathcal{F}(E_1)$ has generators $\ell^{(1)} = (1,1,0,0,0)^{\top}$, $\ell^{(2)} = (0,0,1,1,0)^{\top}$, $\ell^{(3)} = (0,0,1,0,1)^{\top}$. Thus, for $\gamma\in \mathcal{F}(E_1)$, we can write $\gamma=\sum^3_{i=1}\lambda_i\ell^{(i)}=(\lambda_1,\lambda_1,\lambda_2+\lambda_3,\lambda_2,\lambda_3)^\top$, where $\lambda_i\geq 0$.
Note that the reactant matrix $\mathcal{X}_1$ corresponding to $E_1$ can be written as
\begin{align*}
 \begin{pmatrix}
1 & 0 & \theta_1 & \theta_2 & \theta_3 \\
\omega_1 & \omega_2 & 1 & 0 & 0 
\end{pmatrix},
\label{eq:l11b}
\end{align*}
where $\theta_1, \theta_2, \theta_3, \omega_1, \omega_2 \in \{0,1\}$. By the definition of $A(\lambda)$ in \eqref{eq:a}, we have
\begin{align*}
  \det(A(\lambda)) &= \det\left(E_1\operatorname{diag}(\gamma)\mathcal{X}_1^\top \right) \notag \\ 
&=(1-(\theta_2-\theta_1)(\omega_2-\omega_1))\lambda_1\lambda_2+(1-(\theta_3-\theta_1)(\omega_2-\omega_1))\lambda_1\lambda_3.
\end{align*}
Hence, $\det(A(\lambda))\equiv 0$ if and only if $(\theta_2-\theta_1)(\omega_2-\omega_1)=1$ and $(\theta_3-\theta_1)(\omega_2-\omega_1)=1$. This leads directly to two possible cases:
\begin{align*}
\left\{
\begin{aligned}
\theta_1 &= \omega_1 = 1 \\
\theta_2 &= \theta_3 = \omega_2 = 0
\end{aligned}
\right.
\quad \text{or} \quad
\left\{
\begin{aligned}
\theta_1 &= \omega_1 = 0 \\
\theta_2 &= \theta_3 = \omega_2 = 1
\end{aligned}
\right..
\end{align*}
Consequently, if the network $(E_1, {\mathcal X}_1)$ is degenerate, then the reactant matrix  $\mathcal{X}_1$  can be
\begin{align*}
\begin{pmatrix} 1 & 0 & 1 & 0 & 0 \\ 1 & 0 & 1 & 0 & 0 \end{pmatrix}~
\text{or}~
\begin{pmatrix} 1 & 0 & 0 & 1 & 1 \\ 0 & 1 & 1 & 0 & 0 \end{pmatrix}.
\end{align*}
In fact, neither of these two cases can occur. Otherwise, the fourth and fifth columns of $E_1$ would correspond to identical reactions, which contradicts the network structure.  Then, any network corresponding to  $E_1$ is nondegenerate. Similarly, it can be shown that any network corresponding to $E_2$ is also nondegenerate.


 {\it (II)} The proof is similar to that of {\it (I)}.
\end{proof}

\subsection{Networks with \texorpdfstring{$s$}{s} species
\texorpdfstring{$(s\geq 3)$}{(s >= 3)}}
\label{sec4.2}

\begin{lemma}
\label{lm:cl}
    Let $G$ be a two-dimensional zero-one network with $s$ species.   Suppose $G$ is degenerate.   
Assume that the conservation laws are given by $\mathcal{N}_i = a_i\mathcal{N}_1 + b_i\mathcal{N}_2$ for 
$i\in \{3, \ldots, s\}$, where \( a_i,b_i \in \mathbb{R} \).
Let ${\mathcal N}^{1,2}:=\begin{pmatrix} {\mathcal N}_1\\{\mathcal N}_2
\end{pmatrix}
$.
\begin{itemize}
\item[(I)] If 
 ${\mathcal N}^{1,2}$ contains  column vectors from exactly three sets of  
    $B_1$, $B_2$, $B_3$ and $B_4$, 
then $(|a_i|, |b_i|) \in \{(1, 0), (0, 1), (0, 0)\}$. 
\item[(II)] If ${\mathcal N}^{1,2}$ contains column vectors from exactly two sets of  $B_1$, $B_2$, $B_3$ and $B_4$, then 
$(|a_i|, |b_i|) \in \{(1, 0), (0, 1), (0, 0), (1,1)\}$.
\end{itemize}
\end{lemma}
\begin{proof}
{\it (I)} Suppose 
 ${\mathcal N}^{1,2}$ contains  column vectors from exactly three sets of  
    $B_1$, $B_2$, $B_3$ and $B_4$.   Notice that by Lemma \ref{lem2.2}, since $G$ is degenerate, the $2$-species network 
$(\mathcal{N}^{1,2}, \mathcal{X}^{1,2})$ is also degenerate.
 So, by Lemma \ref{lm:b1b2}, the column vectors in $\mathcal{N}^{1,2}$ cannot be from both  $B_1$ and $B_2$. 
Hence, the submatrix 
${\mathcal N}^{1,2}$ contains at least one column vector respectively from $B_1\cup B_2$,
     $B_3$ and $B_4$.
     By \eqref{eq:defb1b2}--\eqref{eq:defb3b4}, 
there exist $j,k,\ell$
such that
\begin{align*}
\operatorname{col}_j({\mathcal N}^{1,2})
=(\varepsilon_j,\varepsilon_j)^\top\;\text{or}\;(\varepsilon_j,-\varepsilon_j)^\top,~ 
\operatorname{col}_k({\mathcal N}^{1,2})
=(\varepsilon_k,0)^\top,~ 
\operatorname{col}_{\ell}({\mathcal N}^{1,2})
=(0,\varepsilon_{\ell})^\top,
\end{align*}
where $\varepsilon_j,\varepsilon_k,\varepsilon_{\ell}\in\{-1,1\}$. 
Since ${\mathcal N}_i=a_i\mathcal{N}_1 + b_i\mathcal{N}_2$, we have
${\mathcal N}_{ik}=a_i\varepsilon_k$ and ${\mathcal N}_{i\ell}=b_i\varepsilon_{\ell}$. Notice that all entries of ${\mathcal N}$ belong to $\{-1, 0, 1\}$. So, we have 
 $|a_i|,  |b_i|\in \{0, 1\}$ and hence $(|a_i|, |b_i|) \in \{(1, 0), (0, 1), (0, 0), (1,1)\}$. Below, we show that $(|a_i|, |b_i|)$ cannot be 
$(1,1)$. 

If
\(
\operatorname{col}_j(\mathcal N^{1,2})
=
(\varepsilon_j,\varepsilon_j)^\top
\in B_1,
\)
then
\(
\mathcal N_{ij}=(a_i+b_i)\varepsilon_j.
\)
Since \(\mathcal N_{ij}\in\{-1,0,1\}\), we have
\(
(a_i,b_i)\notin\{(1,1),(-1,-1)\}.
\) Next, we prove that if $G$ is degenerate, then $(a_i, b_i)\not\in \{(-1,1), (1,-1)\}$. We first assume that $(a_i,b_i)=(-1,1)$ and derive a contradiction.
The case $(a_i,b_i)=(1,-1)$ then follows by swapping  the species
$X_1$ and $X_2$ in $G$. Consider the $3$-species network $G^{1,2,i}$ given by
$$
{\mathcal N}^{1,2,i}:=\begin{pmatrix}
{\mathcal N}_1\\
{\mathcal N}_2\\
{\mathcal N}_i\\
\end{pmatrix}, \;\;
{\mathcal X}^{1,2,i}:=\begin{pmatrix}
{\mathcal X}_1\\
{\mathcal X}_2\\
{\mathcal X}_i\\
\end{pmatrix}.$$
 Since in this case, we have \(\operatorname{col}_j({\mathcal N}^{1,2})
=(\varepsilon_j,\varepsilon_j)^\top\), 
by Theorem~\ref{thm:2s}, we have
\(
\mathcal{X}_1=\mathcal{X}_2.\) As
\(
\mathcal{N}_i
=-\mathcal{N}_1+\mathcal{N}_2,
\)
we obtain
\begin{align}\label{eq:n12i}
\operatorname{col}_k({\mathcal N}^{1,2,i})
=(\varepsilon_k,0,-\varepsilon_k)^\top,~\text{and}~
\operatorname{col}_{\ell}({\mathcal N}^{1,2,i})
=(0,\varepsilon_{\ell},\varepsilon_{\ell})^\top.
\end{align}
If $G$ is degenerate, then
by Lemma~\ref{lem2.2}, the $2$-species network 
$(\mathcal{N}^{2,i}, \mathcal{X}^{2,i})$ is degenerate.
By  \eqref{eq:n12i}, we have $\operatorname{col}_{\ell}({\mathcal N}^{2,i})
=(\varepsilon_{\ell},\varepsilon_{\ell})^\top$. So, by Theorem~\ref{thm:2s}, we have  \(
\mathcal{X}_2=\mathcal{X}_i.
\)
Hence, all three rows in ${\mathcal X}^{1,2,i}$ are the same.  Notice that the $2$-species network 
$(\mathcal{N}^{1,i}, \mathcal{X}^{1,i})$ is also degenerate.
By \eqref{eq:n12i}, we have \(\operatorname{col}_k(\mathcal{N}^{1,i})
=(\varepsilon_k,-\varepsilon_k)^\top\). So, by Theorem~\ref{thm:2s}, we have  
\(
\mathcal{X}_1=
\bar{\mathcal{X}}_i
\), 
which contradicts 
$\mathcal{X}_1=\mathcal{X}_i$.
Therefore,
\(
(a_i,b_i)\neq(-1,1).
\)

If
\(
\operatorname{col}_j(\mathcal N^{1,2})
=
(\varepsilon_j,-\varepsilon_j)^\top
\in B_2,
\)
then
\(
\mathcal N_{ij}=(a_i-b_i)\varepsilon_j.
\)
Since \(\mathcal N_{ij}\in\{-1,0,1\}\), we have
\(
(a_i,b_i)\notin\{(1,-1),(-1,1)\}.
\) Similar to the previous paragraph, we can prove that if $G$ is degenerate, then $(a_i, b_i)\not\in \{(1,1), (-1,-1)\}$.

 {\it (II)} 
 Suppose ${\mathcal N}^{1,2}$ contains column vectors from exactly two sets of  $B_1$, $B_2$, $B_3$ and $B_4$. 
 Recall that by Lemma \ref{lem2.2} and  Lemma \ref{lm:b1b2}, the column vectors in $\mathcal{N}^{1,2}$ cannot be from both  $B_1$ and $B_2$.  Also, notice that if the column vectors in $\mathcal{N}^{1,2}$ are from $B_3$ and $B_4$, then
 similar to the first paragraph of the  proof of {\it (I)}, we can directly show that $(|a_i|, |b_i|) \in \{(1, 0), (0, 1), (0, 0), (1,1)\}$. 

 We only need to consider the cases where the column vectors in
$\mathcal{N}^{1,2}$ are from $B_1$ and $B_3$, or from $B_2$ and $B_3$. If they are from 
 $B_1$ and $B_4$, or from $B_2$ and $B_4$, then we can 
 show that the conclusion holds by simply swapping the two species $X_1$ and $X_2$ in $G$.  
 Notice that $G$ is consistent. If $\mathcal{N}^{1,2}$ contains column vectors from 
 $B_1$ and $B_3$, then 
 by \eqref{eq:defb1b2}--\eqref{eq:defb3b4}, 
there exist $j,k,\ell$
such that
\begin{align*}
\operatorname{col}_j({\mathcal N}^{1,2})
=(1,1)^\top,~ 
\operatorname{col}_k({\mathcal N}^{1,2})
=(-1,-1)^\top,~ 
\operatorname{col}_{\ell}({\mathcal N}^{1,2})
=(\varepsilon_{\ell}, 0)^\top,
\end{align*}
where $\varepsilon_{\ell}\in\{-1,1\}$. 
Since ${\mathcal N}_i=a_i\mathcal{N}_1 + b_i\mathcal{N}_2$, we have ${\mathcal N}_{ij}=a_i+b_i$,
${\mathcal N}_{ik}=-a_i-b_i$ and ${\mathcal N}_{i\ell}=a_i\varepsilon_{\ell}$. Notice that all entries of ${\mathcal N}$ belong to $\{-1, 0, 1\}$. So, we have 
 $(|a_i|, |b_i|) \in \{(1, 0), (0, 1), (0, 0), (1,1)\}$ or $(a_i, b_i) \in \{(1, -2), (-1, 2)\}$. Below, we show that if $G$ is degenerate, then $(a_i, b_i) \not\in \{(1, -2), (-1, 2)\}$. If 
 $(a_i, b_i)=(1, -2)$, then ${\mathcal N}^{1,i}$ contains columns $(1, -1)^\top$ and $(\varepsilon_{\ell}, \varepsilon_{\ell})$. If $(a_i, b_i)=(-1, 2)$, then ${\mathcal N}^{1,i}$ contains columns $(1, 1)^\top$ and $(\varepsilon_{\ell}, -\varepsilon_{\ell})$. In both cases, by Lemma \ref{lm:b1b2},  the $2$-species network 
$(\mathcal N^{1,i}, \mathcal X^{1,i})$ is nondegenerate. Hence, by Lemma \ref{lem2.2}, the network $G$ is also nondegenerate, which is a contradiction. So, we have $(|a_i|, |b_i|) \in \{(1, 0), (0, 1), (0, 0), (1,1)\}$.  Similarly, the conclusion holds if $\mathcal{N}^{1,2}$ contains column vectors from 
 $B_2$ and $B_3$. 
 \end{proof}

\begin{proof}[Proof of Theorem~\ref{thm:deg-char}]
$\Longleftarrow$ If for every $i$, $({\mathcal N}_i, {\mathcal X}_i)$ satisfies \eqref{eq:b.7} or \eqref{eq:b.1}, 
then for any $i,j\in \{1,\ldots, s\}$, the $2$-species network $G^{i,j}$ corresponding to ${\mathcal N}^{i,j}:=\begin{pmatrix} {\mathcal N}_i\\{\mathcal N}_j
\end{pmatrix}$ and ${\mathcal X}^{i,j}:=\begin{pmatrix} {\mathcal X}_i\\{\mathcal X}_j
\end{pmatrix}$ is either one-dimensional or a consistent subnetwork of \eqref{eq:2snet1} or  \eqref{eq:2snet2}. Then, by Theorem \ref{thm:2s} and Lemma \ref{lem2.1}, the network $G$ is degenerate. 

$\Longrightarrow$ Without loss of generality, assume that 
${\mathcal N}_1$ and ${\mathcal N}_2$ are linearly independent and $\mathcal{N}_i = a_i\mathcal{N}_1 + b_i\mathcal{N}_2$ for 
$i\in \{3, \ldots, s\}$. 
If $G$ is degenerate, then by 
Lemma \ref{lem2.2}, the $2$-species network $G^{1,2}$ determined  by ${\mathcal N}^{1,2}:=\begin{pmatrix} {\mathcal N}_1\\{\mathcal N}_2
\end{pmatrix}$ and ${\mathcal X}^{1,2}:=\begin{pmatrix} {\mathcal X}_1\\{\mathcal X}_2
\end{pmatrix}$ is  degenerate.  By Theorem \ref{thm:2s},  the network $G^{1,2}$ is a consistent subnetwork of \eqref{eq:2snet1} or  \eqref{eq:2snet2}. Notice that here, by \eqref{eq:2snet1} and  \eqref{eq:2snet2}, we have ${\mathcal X}_2={\mathcal X}_1$ or ${\mathcal X}_2=\bar{{\mathcal X}}_1$. 
 
If 
 ${\mathcal N}^{1,2}$ contains  column vectors from exactly three sets of  
    $B_1$, $B_2$, $B_3$ and $B_4$, 
then by Lemma \ref{lm:cl} {\it (I)},  $(|a_i|, |b_i|) \in \{(1, 0), (0, 1), (0, 0)\}$. Notice that $G$ contains no trivial species. So, for any $i\in \{3, \ldots, s\}$, we have 
${\mathcal N}_i\in \{\pm{\mathcal N}_1, \pm{\mathcal N}_2\}$. Definitely, one of 
the $2$-species networks $G^{1,i}$ and $G^{2,i}$ is two-dimensional. Without loss of generality, assume that $G^{1,i}$ is two-dimensional. In this case, ${\mathcal N}_i=\pm{\mathcal N}_2$. Notice that $G^{1,i}$ is degenerate by Lemma \ref{lem2.2}. So, by Theorem \ref{thm:2s},  the network $G^{1,i}$ is a consistent subnetwork of \eqref{eq:2snet1} or  \eqref{eq:2snet2}. Hence, 
${\mathcal X}_i$ can only be ${\mathcal X}_2$ or $\bar{\mathcal{X}}_2$. Finally, notice that the network \eqref{eq:mainnet1} stated in {\it (I)} is exactly the network \eqref{eq:2snet1}. Therefore, statement {\it (I)} holds. 

Suppose ${\mathcal N}^{1,2}$ contains column vectors from exactly two sets of  $B_1$, $B_2$, $B_3$ and $B_4$.  Notice that  ${\mathcal N}^{1,2}$ cannot be formed by the column vectors from $B_1$ and $B_2$ by Lemma \ref{lm:b1b2}. Below, we prove the statement {\it (II)} by discussing the following Cases 1--3. Also, notice that 
the only consistent subnetwork of \eqref{eq:mainnet2} is itself. So, ${\mathcal E}_i$ and ${\mathcal Y}$ stated in {\it (II)} satisfy  
\begin{align}
\label{eq:s6r4}
\left(
\begin{array}{c}
\mathcal{E}_1\\
\mathcal{E}_2 \\
\mathcal{E}_3 
\end{array}\right)=\left(
\begin{array}{cccc}
-1 & 1 & 0 & 0 \\
0 & 0 & -1 & 1 \\
-1 & 1 & -1 & 1 
\end{array}\right),~\text{and}~\mathcal{Y}=\left(
\begin{array}{cccc}
1 & 0 & 1 & 0 
\end{array}\right). 
\end{align}

\begin{itemize}
\item[Case 1.]If ${\mathcal N}^{1,2}$ contains column vectors from exactly  $B_3$ and $B_4$, then by Theorem \ref{thm:2s},  
the degenerate $2$-species network $G^{1,2}$ must be  
one of the following networks:
\begin{align}
\begin{pmatrix} {\mathcal E}_1\\{\mathcal E}_2
\end{pmatrix}=\left(
\begin{array}{cccccc}
-1 & 1 & 0 & 0\\
 0 & 0 &-1 & 1
\end{array}\right), ~\begin{pmatrix} {\mathcal Y}\\{\mathcal Y}
\end{pmatrix}=\left(
\begin{array}{cccccc}
1 & 0 & 1 & 0 \\
1 & 0 & 1 & 0 
\end{array}\right);\label{eq:b3b4net1}
\end{align}
\begin{align}
\begin{pmatrix} {\mathcal E}_1\\-{\mathcal E}_2
\end{pmatrix}=\left(
\begin{array}{cccccc}
 -1 & 1 & 0 & 0\\
 0 & 0 & 1 & -1
\end{array}\right), ~\begin{pmatrix} {\mathcal Y}\\{\bar{\mathcal Y}}
\end{pmatrix}=\left(
\begin{array}{cccccc}
 1 & 0 & 1 & 0\\
 0 & 1 & 0 & 1
\end{array}\right),\label{eq:b3b4net2}
\end{align}
where \eqref{eq:b3b4net1} and \eqref{eq:b3b4net2} are, respectively, the unique consistent subnetworks of \eqref{eq:2snet1} and \eqref{eq:2snet2} that have stoichiometric matrices formed by the columns from $B_3$ and $B_4$. By Lemma \ref{lm:cl}  {\it (II)}, 
$(|a_i|, |b_i|) \in \{(1, 0), (0, 1), (0, 0), (1,1)\}$. Observe that regardless of whether 
 $G^{1,2}$ is \eqref{eq:b3b4net1} or \eqref{eq:b3b4net2}, we have 
${\mathcal N}_i\in \{\pm {\mathcal E}_1, \pm {\mathcal E}_2, \pm {\mathcal E}_3, {\mathcal E}_1-{\mathcal E}_2,  {\mathcal E}_2-{\mathcal E}_1\}$ (from  \eqref{eq:s6r4}, ${\mathcal E}_3={\mathcal E}_1+{\mathcal E}_2$). Then, at least one of 
the $2$-species networks $G^{1,i}$ and $G^{2,i}$ is two-dimensional, and by  Lemma \ref{lem2.2}, the two-dimensional one must be degenerate.  So, 
if ${\mathcal N}_i\in \{\pm {\mathcal E}_1, \pm {\mathcal E}_2, \pm {\mathcal E}_3\}$, then by Theorem \ref{thm:2s}, we have \eqref{eq:b.1}. In the rest of the proof, we show that 
${\mathcal N}_i\not\in \{{\mathcal E}_1-{\mathcal E}_2,  {\mathcal E}_2-{\mathcal E}_1\}$. If ${\mathcal N}_i={\mathcal E}_1-{\mathcal E}_2$, then both $G^{1,i}$ and $G^{2,i}$  are two-dimensional and degenerate. 
If $G^{1,2}$ is \eqref{eq:b3b4net1}, then ${\mathcal N}_k={\mathcal E}_k$ and ${\mathcal X}_k={\mathcal Y}$ for $k=1,2$. 
Since ${\mathcal N}^{1,i}$ contains the column vectors from $B_1$, 
by 
Theorem \ref{thm:2s}, $G^{1,i}$ is a subnetwork of \eqref{eq:2snet1} and so 
${\mathcal X}_i={\mathcal X}_1={\mathcal Y}$. 
However, since ${\mathcal N}^{2,i}$ contains the column vectors from $B_2$, 
by 
Theorem \ref{thm:2s}, $G^{2,i}$ is a subnetwork of \eqref{eq:2snet2} and so 
${\mathcal X}_i=\bar{{\mathcal X}_2}=\bar{{\mathcal Y}}$, which is a contradiction.  
If $ G^{1,2}$ is \eqref{eq:b3b4net2}, then by
Theorem \ref{thm:2s}, $G^{1,i}$ is a subnetwork of \eqref{eq:2snet1} and so 
${\mathcal X}_i={\mathcal X}_1={\mathcal Y}$. 
However, since ${\mathcal N}^{2,i}$ also contains the column vectors from $B_1$, 
by 
Theorem \ref{thm:2s}, $G^{2,i}$ is a subnetwork of \eqref{eq:2snet1} and so 
${\mathcal X}_i={\mathcal X}_2=\bar{{\mathcal Y}}$, which is  again a contradiction. So, ${\mathcal N}_i\neq {\mathcal E}_1-{\mathcal E}_2$. 
Similarly, one can show that ${\mathcal N}_i\neq {\mathcal E}_2-{\mathcal E}_1$.

\item[Case 2.]If ${\mathcal N}^{1,2}$ contains column vectors from exactly  $B_1$ and $B_3$ or from exactly $B_1$ and $B_4$, then by Theorem \ref{thm:2s},  
the degenerate $2$-species network $G^{1,2}$ is
\begin{align}
\begin{pmatrix}
\mathcal E_3\\
\mathcal E_1
\end{pmatrix}
=
\begin{pmatrix}
-1&1&-1&1\\
-1&1&0&0
\end{pmatrix},~
\begin{pmatrix}
\mathcal Y\\
\mathcal Y
\end{pmatrix}
=
\begin{pmatrix}
1&0&1&0\\
1&0&1&0
\end{pmatrix},
\label{eq:b1b3b1b4}
\end{align}
where \eqref{eq:b1b3b1b4} is the unique consistent subnetwork of \eqref{eq:2snet1} that has the stoichiometric matrix formed by the columns from exactly  $B_1$ and $B_3$ or from exactly $B_1$ and $B_4$. By Lemma \ref{lm:cl}  {\it (II)}, 
$(|a_i|, |b_i|) \in \{(1, 0), (0, 1), (0, 0), (1,1)\}$. Notice that $G$ contains no trivial species. So, for any $i\in \{3, \ldots, s\}$, we have ${\mathcal N}_i\in \{\pm {\mathcal E}_1, \pm {\mathcal E}_2, \pm {\mathcal E}_3, {-\mathcal E}_1-{\mathcal E}_3,{\mathcal E}_1+{\mathcal E}_3  \}$ (from  \eqref{eq:s6r4}, ${\mathcal E}_2={\mathcal E}_3-{\mathcal E}_1$). Since all entries of ${\mathcal N}$ belong to $\{-1, 0, 1\}$, we have $\mathcal N_i\neq{-\mathcal E}_1-{\mathcal E}_3$  and $\mathcal N_i\neq{\mathcal E}_1+{\mathcal E}_3$. Therefore,
$\mathcal N_i
\in
\{\pm\mathcal E_1,\pm\mathcal E_2,\pm\mathcal E_3\}$, and by Theorem \ref{thm:2s}, we have \eqref{eq:b.1}.
\item[Case 3.]If ${\mathcal N}^{1,2}$ contains column vectors from exactly  $B_2$ and $B_3$ or from exactly $B_2$ and $B_4$, then by Theorem \ref{thm:2s},  
the degenerate $2$-species network $G^{1,2}$ is
\begin{align}
\begin{pmatrix}
\mathcal E_3\\
-\mathcal E_1
\end{pmatrix}
=
\begin{pmatrix}
-1&1&-1&1\\
1&-1&0&0
\end{pmatrix},~
\begin{pmatrix}
\mathcal Y\\
\mathcal Y
\end{pmatrix}
=
\begin{pmatrix}
1&0&1&0\\
0&1&0&1
\end{pmatrix},
\label{eq:b2b3b2b4}
\end{align}
where \eqref{eq:b2b3b2b4} is the unique consistent subnetwork of \eqref{eq:2snet2} that has the stoichiometric matrix formed by the columns from exactly  $B_2$ and $B_3$ or from exactly $B_2$ and $B_4$. By Lemma \ref{lm:cl}  {\it (II)}, 
$(|a_i|, |b_i|) \in \{(1, 0), (0, 1), (0, 0), (1,1)\}$. Notice that $G$ contains no trivial species. So, for any $i\in \{3, \ldots, s\}$, we have ${\mathcal N}_i\in \{\pm {\mathcal E}_1, \pm {\mathcal E}_2, \pm {\mathcal E}_3, {-\mathcal E}_1-{\mathcal E}_3,{\mathcal E}_1+{\mathcal E}_3  \}$ (from  \eqref{eq:s6r4}, ${\mathcal E}_2={\mathcal E}_3-{\mathcal E}_1$). Since all entries of ${\mathcal N}$ belong to $\{-1, 0, 1\}$, we have $\mathcal N_i\neq{-\mathcal E}_1-{\mathcal E}_3$  and $\mathcal N_i\neq{\mathcal E}_1+{\mathcal E}_3$. Therefore,
$\mathcal N_i
\in
\{\pm\mathcal E_1,\pm\mathcal E_2,\pm\mathcal E_3\}$, and by Theorem \ref{thm:2s}, we have \eqref{eq:b.1}.
\end{itemize}
\end{proof}

\section{Discussion}
\label{sec5}
In this paper, we obtained a complete characterization of degenerate two-dimensional zero-one networks without trivial species. Our main results provide equivalent network-theoretic and matrix-theoretic descriptions (Theorems~\ref{thm:main} and \ref{thm:deg-char}), showing that degeneracy in this class is entirely governed by two prototypical structures. As a consequence, every such degenerate network exhibits binomial steady-state equations, revealing an intrinsic connection between degeneracy and binomiality.

The proof framework developed here extends naturally to higher-dimensional zero-one networks. However, the explicit classification of degenerate cases in dimensions greater than two becomes substantially more involved, as the number of candidate row patterns grows rapidly with the dimension, leading to a significantly larger computational burden in the combinatorial analysis.



\bibliographystyle{siamplain}
\bibliography{references}
\end{document}